\documentclass[10pt,twocolumn,letterpaper]{article}
\usepackage{amsmath}
\usepackage[utf8]{inputenc}
\usepackage[margin=1in]{geometry}
\usepackage[T1]{fontenc}
\usepackage{graphicx}
\usepackage{booktabs,amsmath, bbm}
\usepackage{float}
\usepackage{amsthm}
\usepackage{amssymb}
\usepackage{upquote}
\usepackage{pseudocode}
\usepackage{subfig}
\usepackage{tikz}
\usepackage{algorithm}
\usepackage{algpseudocode}
\usepackage{parskip}
\usepackage{scrextend}
\usepackage{fancyhdr}
\usepackage{pgfplots}
\usepackage{mathtools}
\usepackage[most]{tcolorbox}
\usepackage{xparse}
\usepackage{bm}

\usetikzlibrary{arrows.meta}

\tikzset{
    myarrow/.style={-{Triangle[length=3mm,width=1mm]}}
}

\definecolor{darkpastelgreen}{rgb}{0.01, 0.75, 0.24}

\newtheorem{innercustomgeneric}{\customgenericname}
\providecommand{\customgenericname}{}
\newcommand{\newcustomtheorem}[2]{%
  \newenvironment{#1}[1]
  {%
   \renewcommand\customgenericname{#2}%
   \renewcommand\theinnercustomgeneric{##1}%
   \innercustomgeneric
  }
  {\endinnercustomgeneric}
}

\newcustomtheorem{ntheorem}{Theorem}
\newcustomtheorem{nlemma}{Lemma}
\newcustomtheorem{nex}{Exercise}

\newtheorem{theorem}{Theorem}
\newtheorem{corollary}[theorem]{Corollary}
\newtheorem{lemma}[theorem]{Lemma}

\DeclarePairedDelimiterX{\innerprod}[2]{\langle}{\rangle}{#1, #2}

\newcommand{\norm}[2]{\left\| #1 \right\|_{#2}}
\newcommand{\grad}{\nabla}

\newcommand{\R}{\mathbb{R}}

\newcommand{\C}{\mathbb{C}}

\newcommand{\E}{\mathbb{E}}
\newcommand{\Cov}{\mathrm{Cov}}
\newcommand{\mc}[1]{\mathcal{#1}}

\newcommand{\mbf}[1]{\mathbf{#1}}
\newcommand{\mbx}{\mathbf{x}}
\newcommand{\mby}{\mathbf{y}}
\newcommand{\mbA}{\mathbf{A}}
\newcommand{\mbC}{\mathbf{C}}

\newcommand{\tweedie}{\bm{\mu}_{0 \mid t}}

\newcommand{\Ours}{FGPS}
\newcommand{\score}[1]{\grad_{\mbx_t} \log #1}

\usepackage[pagenumbers]{iccv} % To force page numbers, e.g. for an arXiv version

\algrenewcommand\algorithmiccomment[1]{\hfill \textcolor{gray}{// #1}}
\definecolor{iccvblue}{rgb}{0.21,0.49,0.74}
\usepackage[pagebackref,breaklinks,colorlinks,allcolors=iccvblue]{hyperref}

\def\paperID{8931} % *** Enter the Paper ID here
\def\confName{ICCV}
\def\confYear{2025}

\title{Frequency-Guided Posterior Sampling for Diffusion-Based Image Restoration}

\author{Darshan Thaker \\
University of Pennsylvania \\
{\tt\small dbthaker@seas.upenn.edu}
\and
Abhishek Goyal \\
University of Pennsylvania \\
{\tt\small abhi2358@seas.upenn.edu}
\and
Ren\'e Vidal \\
University of Pennsylvania \\
{\tt\small vidalr@seas.upenn.edu}
}

\begin{document}
\maketitle
\begin{abstract}
    Image restoration aims to recover high-quality images from degraded observations. When the degradation process is known, the recovery problem can be formulated as an inverse problem, and in a Bayesian context, the goal is to sample a clean reconstruction given the degraded observation. Recently, modern pretrained diffusion models have been used for image restoration by modifying their sampling procedure to account for the degradation process. However, these methods often rely on certain approximations that can lead to significant errors and compromised sample quality. In this paper, we propose a simple modification to existing diffusion-based restoration methods that exploits the frequency structure of the reverse diffusion process. Specifically, our approach, denoted as Frequency Guided Posterior Sampling (FGPS), introduces a time-varying low-pass filter in the frequency domain of the measurements, progressively incorporating higher frequencies during the restoration process. We provide the first rigorous analysis of the approximation error of FGPS for linear inverse problems under distributional assumptions on the space of natural images, demonstrating cases where previous works can fail dramatically. On real-world data, we develop an adaptive curriculum for our method's frequency schedule based on the underlying data distribution. FGPS significantly improves performance on challenging image restoration tasks including motion deblurring and image dehazing. The source code is available at \url{https://github.com/darshanthaker/FGPS}.
\end{abstract}    
\section{Introduction}
\label{sec:intro}

Image restoration is a fundamental problem in computer vision whose goal is to recover high-quality images from degraded observations. Many image restoration problems can be cast as inverse problems, where the task is to recover an underlying signal $\mbx_0 \in \R^n$ from noisy measurements $\mby \in \R^m$ that are obtained from $\mbx_0$ via
\begin{equation} \label{eq:forward}
    \mby = \mc{A}(\mbx_0) + \mbf{z}.
\end{equation}
In the above, $\mc{A}: \R^n \to \R^m$ denotes the forward operator, which we assume to be known, and $\mbf{z} \in \R^m$ is additive Gaussian noise.  
In most applications, recovering $\mbx_0$ is difficult because the forward operator is a many-to-one mapping. To address this challenge, researchers incorporate prior knowledge about $\mbx_0$, including sparsity \cite{candes2006stable, donoho2005stable}, low-rank structure \cite{candes2011robust, candes2012exact}, total variation \cite{rudin1992nonlinear}, or deep generative priors \cite{pan2021exploiting, bora2017compressed}. One way of enforcing this knowledge is by solving a regularized optimization problem where the objective is to find an $\mbx_0$ that maximizes the likelihood under the chosen prior model while satisfying the constraint $\mby = \mc{A}(\mbx_0)$. Another approach is to follow the Bayesian viewpoint where instead of recovering an exact $\mbx_0$, we sample from $p(\mbx \mid \mby)$ where \Cref{eq:forward} is the model for $p(\mby \mid \mbx_0)$ \cite{stuart2010inverse}. The advantage of this interpretation is that we can draw multiple plausible samples. In image restoration tasks, the diversity of samples can help the user pick a visually appealing reconstruction.

Recently, diffusion models \cite{song2019generative, ho2020denoising} have shown tremendous success in modeling complex probability distributions, allowing one to generate diverse and visually coherent samples in image, video, and audio tasks \cite{dhariwal2021diffusion,ho2022video,kong2020diffwave}. During training, diffusion models learn to reverse a gradual noising process that transforms data $\mbx_0$ to pure noise $\mbx_T$ over $T$ steps, where $\mbx_t$ denotes an intermediate step of the reverse diffusion process. At test time, they generate samples by iteratively denoising random noise, progressively refining it into coherent data. Naturally, these models can also be used to develop samplers for $p(\mbx \mid \mby)$ in inverse problem settings. However, while one can train a problem-specific diffusion model to sample directly from $p(\mbx \mid \mby)$ \cite{saharia2022palette}, this requires training new diffusion models for every inverse problem, limiting the practical use of such approaches. 

A popular approach which does not require training a new diffusion model for $p(\mbx \mid \mby)$ is to modify the sampling process of a pretrained unconditional diffusion model so that it produces samples that are consistent with the given measurement $\mby$ \cite{chung2022diffusion, kawar2022denoising, song2020score, song2023pseudoinverse, wang2022zero}. These methods are based on decomposing the conditional score function $\grad_{\mbx_t} \log p(\mbx_t \mid \mby)$ as $\grad_{\mbx_t} \log p(\mbx_t) + \grad_{\mbx_t} \log p(\mby \mid \mbx_t)$ using the Bayes rule. While the first term is learned by a pretrained unconditional diffusion model, the second term is analytically intractable in general because $\mbx_t$ is the result of a complex denoising process. Thus, in practice, we use approximations for this term, typically by first computing $\bm{\mu}_{0 \mid t} = \E[\mbx_0 \mid \mbx_t]$ using the learned diffusion model, known as the Tweedie estimate \cite{efron2011tweedie}, and following the gradient of a loss that measures the deviation between $\mc{A}(\bm{\mu}_{0 \mid t})$ and $\mby$. While intuitive, the approximation error from the Tweedie estimate is difficult to control. Previous work has attempted to address this by designing heuristics that ensure that the modifications to $\mbx_t$ respect the intermediate data manifold learned by the diffusion model. However, these approaches make strong assumptions about the data manifold's structure, such as linearity, or require that $\tweedie$ remains on the clean data manifold throughout all iterations \cite{chung2022improving, yang2024guidance}.%Most commonly, this is performed by scaling the gradient term. However, if the scale is too large, the samples $\mbx_t$ may deviate far from the noisy data manifold, and the predicted $\bm{\mu}_{0 \mid t}$ by the pretrained diffusion model in the following steps may be of low quality. If the modification is too small, the samples are unlikely to be consistent with the measurement $\mby$. 

Our work aims to better understand and control this approximation error by analyzing images in the frequency domain. Our main contributions are:

\begin{enumerate}
    \item We develop a new unsupervised diffusion-based image restoration algorithm, denoted as \textbf{F}requency \textbf{G}uided \textbf{P}osterior \textbf{S}ampling (\Ours). Our algorithm leverages the observation that diffusion models generate images hierarchically in the frequency domain \cite{rissanen2022generative}. Thus, to control the approximation error from estimating $\grad_{\mbx_t} \log p(\mby \mid \mbx_t)$, we apply a time-varying low-pass filter to $\mby$ which progressively includes higher frequency ranges during the reverse diffusion process.
    
    \item We then perform a precise theoretical analysis of the approximation errors in posterior sampling methods for linear inverse problems under distributional assumptions on the data. Next, exploiting the frequency distribution of natural images, we give precise examples of forward operators where FGPS has small approximation error compared to prior work.

    \item In practice, we develop a data-dependent curriculum for the low-pass filters which allows the reverse process to adapt to spectral properties of the data distribution. While being easy to implement, our algorithm leads to a significant boost in performance on challenging linear and nonlinear image restoration problems, such as motion deblurring and image dehazing. 
\end{enumerate}

\section{Background}
\label{sec:background}

\subsection{Diffusion Models} \label{sec:background_dm}
Diffusion models are a class of score-based generative models that sample from a given data distribution by learning to reverse a gradual noising process \cite{ho2020denoising, song2019generative}. This noising process, known as the forward process, is represented in continuous time as an It\^{o} Stochastic Differential Equation (SDE) for $t \in [0, 1]$
\begin{equation} \label{eq:forward_sde}
        d \mbx_t = -\frac{1}{2} \beta(t) \mbx_t dt + \sqrt{\beta(t)} d \mbf{w}_t. 
\end{equation}
The SDE is characterized by its drift coefficient $-\frac{1}{2} \beta(t)$ and diffusion coefficient $\sqrt{\beta(t)}$, with $\mbf{w}_t$ representing the standard Wiener process over the interval $[0, 1]$. The initial condition of this SDE is a datapoint $\mbx_0 \sim p(\mbx)$. The drift and diffusion coefficients are chosen such that $\mbx_t$ goes from the clean data distribution at $t = 0$ to pure noise $\mbx_1 \sim \mc{N}(\mbf{0}, \mbf{I})$ at $t = 1$. A special case is known as the Variance-Preserving SDE (VP-SDE), where we can write the transition density $p_t (\mbx_t \mid \mbx_0)$ as
\begin{equation} \label{eq:forward_process}
    p(\mbx_t \mid \mbx_0) = \mc{N}(\sqrt{\bar{\alpha}_t} \mbx_0, (1 - \bar{\alpha}_t) \mathbf{I})\,,
\end{equation}
% For example, when $\beta(t) = 2$, we recover the Ornstein-Uhlenbeck process where $\alpha(t) = \exp(-2t)$. 
for $\bar{\alpha}_t = e^{ \int_0^t \beta(s) \ ds}$ \cite{ho2020denoising}. 

The reverse process aims to remove the noise from $\mbx_1$ and draw a sample from $\mbx_0$ following the reverse-time SDE
\begin{equation} \label{eq:reverse_sde}
    d \mbx_t = \left[-\frac{\beta(t)}{2} \mbx_t - \beta(t) \grad_{\mbx_t} \log p_t(\mbx_t) \right] dt + \sqrt{\beta(t)} d \bar{\mathbf{w}}_t\,,
\end{equation}
where $t$ and the new Wiener process $\bar{\mathbf{w}}_t$ both run in reverse. A seminal result of Anderson gives that this reverse process has the same probability distribution as the forward process for every $t$ \cite{anderson1982reverse}. In order to estimate the time-dependent score function $\grad_{\mbx_t} \log p_t(\mbx_t)$, diffusion models use the denoising score matching technique to train a neural network $s_\theta(\cdot)$ that predicts the noise added to a datapoint, which approximates the score function \cite{vincent2011connection}. Sampling from $p(\mbx)$ is performed by solving a discrete-time version of \Cref{eq:reverse_sde} over $T$ steps, where $\mbx_T \sim \mc{N}(\mbf{0}, \mbf{I})$ \cite{ho2020denoising,song2020score}. 

\subsection{Posterior Sampling} \label{sec:background_ps}

In inverse problems, we are given measurements $\mby$ as in \Cref{eq:forward}, and the goal is to recover the underlying clean $\mbx_0$ given $\mby$. Throughout the paper, we will refer to $\mc{A}$ as the known forward operator. Because $\mc{A}$ is not invertible in general, the inverse problem can be highly ill-posed, necessitating data modelling assumptions. 

A popular Bayesian strategy to solve the problem is to develop samplers for $p(\mbx \mid \mby)$ assuming we have a prior for $p(\mbx)$ modeled by a pretrained generative model such as a diffusion model. For instance we can extend \Cref{eq:reverse_sde} to sample from $p(\mbx \mid \mby)$ by using a conditional score function
\begin{equation} \label{eq:reverse_cond_sde}
    d \mbx_t = \left[-\frac{\beta(t)}{2} \mbx_t - \beta(t) \grad_{\mbx_t} \log p_t(\mbx_t \mid \mby) \right] dt + \sqrt{\beta(t)} d \bar{\mathbf{w}}_t.
\end{equation}
While one could train a conditional diffusion model to approximate the conditional score, this would require retraining models for every inverse problem we wish to solve \cite{saharia2022palette}. Instead, we use Bayes rule to write the conditional score as
\begin{equation} \label{eq:bayes_rule}
     \grad_{\mbx_t} \log p(\mbx_t \mid \mby) = \underbrace{\grad_{\mbx_t} \log p(\mbx_t)}_{\text{unconditional score}} + \underbrace{\grad_{\mbx_t} \log p(\mby \mid \mbx_t)}_{\text{noisy likelihood score}}
\end{equation}
While a pretrained diffusion model effectively approximates the unconditional score, the noisy likelihood score remains intractable due to its factorization
\begin{equation} \label{eq:noisy_factorization}
    p(\mby \mid \mbx_t) = \int p(\mby \mid \mbx_0) p(\mbx_0 \mid \mbx_t) d \mbx_0,
\end{equation}
where $p(\mbx_0 \mid \mbx_t)$ represents the complex distribution of denoised estimates of $\mbx_t$. Even though $p(\mbx_0 \mid \mbx_t)$ can be complex, the Tweedie's estimate gives us the posterior mean $\tweedie = \E[ \mbx_0 \mid \mbx_t]$ as a function of the unconditional score function \cite{efron2011tweedie} (refer to Appendix \ref{sec:app:background} for more background). Thus, some unsupervised methods approximate $p(\mbx_0 \mid \mbx_t)$ using $\tweedie$, for example as a Dirac delta around $\tweedie$ \cite{chung2022diffusion} or an isotropic Gaussian distribution with mean $\tweedie$ \cite{song2023pseudoinverse}. This results in tractable approximations for the conditional score $\grad_{\mbx_t} \log p(\mbx_t \mid \mby)$, but introduces a potentially large approximation error. In this work, we propose a novel approximation for the conditional score and quantify its approximation gap exactly under distributional assumptions on $\mbx_0$.

\subsection{Signal Processing Basics}

Our method relies on the frequency decomposition of natural images, so in this section, we review the terminology used for the remainder of the paper. The Discrete Fourier Transform (DFT) of a finite signal $\mbx \in \R^n$, denoted as $\mc{F}(\mbx)$, maps the signal to the frequency domain. At each DFT sample frequency $f_k = \frac{k}{n}$ for $k = 1, \dots, n - 1$, the DFT value $\mc{F}(\mbx)[f_k]$ represents the complex-valued frequency component in the signal. We define the signal's \textit{power spectral density} as the expected squared magnitude (the power) of the signal at different sample frequencies
\begin{equation} \label{eq:psd}
    S(f_k) = \E_\mbx \left[\frac{| \mc{F}(\mbx)[f_k] |^2}{n} \right].
\end{equation}
When $S(f_k) = c |f_k|^{-\beta}$ for constants $c, \beta > 0$, we say that the signal follows a power law in the frequency domain with parameters $c, \beta$. %It is well-known that natural images tend to follow a radially averaged power law in the frequency domain, such that lower frequencies dominate the signal power \cite{van1996modelling}. We leverage this crucial property in our work to develop a new method for solving inverse problems.

\section{Related Work}
\label{sec:related_work}

Inverse problems of the form in \Cref{eq:forward} have been solved for many years with different priors, such as sparse/low-rank priors \cite{candes2006stable, candes2011robust, candes2012exact} or generative priors \cite{bora2017compressed, pan2021exploiting}. In this work, we focus on posterior sampling methods that use pretrained diffusion models and approximate the conditional score to perform guidance \cite{chung2022diffusion, song2023pseudoinverse, wang2022zero, rozet2024learning, song2020score, choi2021ilvr, ho2022classifier}. These have been applied to inverse problems such as image deblurring, super-resolution \cite{daniels2021score}, and image inpainting \cite{lugmayr2022repaint, chung2022diffusion}. However, they can incur significant approximation errors as we will illustrate in \Cref{sec:limitations}. Works such as \cite{chung2022improving, yang2024guidance} propose methods to mitigate these approximation errors, but rely on strong assumptions on the data manifold. In this work, we use a more practical assumption that the data obeys a power law in the frequency domain \cite{van1996modelling} in order to develop a new approximation. %Beyond the scope of inverse problems, we note that there are several guidance methods for diffusion models, typically using classifiers \cite{ho2020denoising, ho2022classifier} or handling other conditional generation tasks such as text-guided diffusion \cite{zhang2023adding, mou2024t2i}. 

A line of work has studied diffusion models in a transformed space to perform guidance for inverse problems \cite{kawar2021snips, kawar2022denoising}. This is in fact closely related to our work as our work is akin to performing diffusion in the Fourier domain instead, similar to \cite{lv2024fourier}. Building on recent connections between diffusion models and spectral autoregression \cite{dieleman2024spectral}, our work differs from previous approaches \cite{song2020score, choi2021ilvr} by using a frequency-guided measurement schedule rather than one based on the unconditional diffusion model's variance schedule. We refer the reader to \cite{daras2024survey} for a comprehensive survey on using diffusion models for inverse problems.

\section{Frequency Guided Posterior Sampling} \label{sec:method}

As illustrated in Section \ref{sec:background_ps}, the main challenge in developing unsupervised inverse problem solvers is estimating the noisy likelihood score $\score{p(\mby \mid \mbx_t)}$, which depends on the intractable denoising distribution $p(\mbx_0 \mid \mbx_t)$. It is common to substitute this denoising distribution with an approximation that depends only on the posterior mean $\tweedie = \E[\mbx_0 \mid \mbx_t]$, which is tractable using Tweedie's formula. For example, Diffusion Posterior Sampling (DPS) approximates the noisy likelihood score with $\score{p(\mby \mid \tweedie)}$ \cite{chung2022diffusion}. However, these methods can incur a significant approximation error by using simple distributions centered around the posterior mean instead of the true distribution.  

The key challenge is ensuring that $\mbx_t$ stays on the data manifold learned by the diffusion model even after it is updated using an approximation of the gradient $\score{p(\mby \mid \mbx_t)}$. Prior approaches address this by assuming that the data manifold is linear or that $\tweedie$ remains on the clean data manifold \cite{chung2022improving, yang2024guidance}. However, these assumptions are often violated in practice. Our main insight is to use the observation that diffusion models generate images hierarchically. Specifically, recent work has connected diffusion models and inverse heat dissipation, concluding that diffusion models naturally generate images from low frequencies (corresponding to coarse features) to high frequencies (corresponding to fine-grained features) \cite{rissanen2022generative} \footnote{This effect is because the forward process noise schedule removes high frequency components proportional to the variance of the noise.}. If there is a mismatch in frequency content between $\tweedie$ and the measurement $\mby$, updating $\mbx_t$ by following an approximation of $\score{p(\mby \mid \mbx_t)}$ can move $\mbx_t$ off the learned data manifold. This problem is exacerbated during the early stages of the reverse process since the early stages play a crucial role in ensuring the generated output matches the fundamental low-frequency structure of the measurements.

%While the results from Figures \ref{fig:highpass_limitations} and \ref{fig:approx_gap} might appear artificial, this phenomenon frequently emerges in complex real-world image restoration problems. For example, in directional motion deblurring, the forward operator does indeed act as a high-pass filter in certain spatial directions of the image (see \Cref{sec:app:further_motivation}). Unfortunately, for real image data that follows a complex non-Gaussian distribution, we cannot exactly compute the true conditional score as we did in \Cref{thm:approx_gap}. Despite this, we hope to derive practical insights from our theoretical analysis because the underlying data distribution in our synthetic experiments mimics the frequency distribution of natural images. 

%Specifically, \Cref{fig:approx_gap} reveals a crucial insight: when using a low-pass filter as the forward operator, the approximation gap remains small during the early stages of the reverse process. Diffusion models inherently generate images hierarchically, from coarse to fine details, so the early stages of the reverse process play a crucial role in establishing the fundamental structure that matches the measurements. If the approximation gap is large in this regime, as is the case when the forward operator is a high-pass filter, it is unlikely the sampling procedure will produce a good reconstruction.

We propose a simple fix to this issue by retaining only low frequency components of the measurements at the initial steps of the reverse process and progressively adding high frequency components in later steps. To implement this, we convolve the original measurement $\mby$ with a time-varying low-pass filter $\mbf{k}_t$. Denoting $\phi_t(\mby)$ as this convolution, the approximation we propose for the noisy likelihood score is 
\begin{equation}
    \score{p(\mby \mid \mbx_t)} \approx \score{p(\phi_t(\mby) \mid \tweedie)}\\,
\end{equation}
where $\tweedie$ denotes the Tweedie estimate $\E[\mbx_0 \mid \mbx_t]$ predicted by the diffusion model. In order to compute the above approximation, we need to compute the model likelihood for the new measurements. Since convolution of a signal $\mby$ with a low pass filter can be represented as a circulant matrix multiplied by $\mby$, we have that $\phi_t(\mby) = \mbC_t \mby$ where $\mbC_t$ is a time-dependent circulant matrix representing $\mbf{k}_t$. This circulant matrix has its first row as the flattened $\mbf{k}_t$ and all other rows as cyclic shifts. Thus, using properties of the multivariate normal distribution, we have that the updated model likelihood is 
\begin{equation} \label{eq:new_likelihood}
    p(\phi_t(\mby) \mid \mbx_0) = \mc{N}(\mbC_t \mc{A}(\mbx_0), \sigma_y^2 \mbC_t\mbC_t^T)\\.
\end{equation}
Given this likelihood, our approximation of the noisy likelihood score becomes
\begin{equation}
    p(\phi_t(\mby) \mid \tweedie) = \mc{N}(\mbC_t \mc{A}(\tweedie), \sigma_y^2 \mbC_t\mbC_t^T)\\, 
\end{equation}
and its corresponding gradient is 
\begin{equation} \label{eq:our_score_approx}
    \score{p(\phi_t(\mby) \mid \tweedie)} = \mbf{S}_t \grad_{\mbx_t} \norm{\mbC_t\mby - \mbC_t \mc{A}(\tweedie)}{2}^2\\, 
\end{equation}
where $\mbf{S}_t = (\sigma_y^2 \mbC_t \mbC_t^T)^{-1}$. Since $\tweedie$ can be computed using one network evaluation of the diffusion model, its gradient with respect to $\mbx_t$ can be computed via backpropagation, as in DPS. Alternating unconditional sampling and conditional sampling using their corresponding score functions gives us our algorithm, which we call \textbf{F}requency \textbf{G}uided \textbf{P}osterior \textbf{S}ampling (\Ours) outlined in Algorithm \ref{alg:ours}.

\noindent \textbf{Dataset-informed choice of $\mbf{k}_t$}. An important choice for our method is choosing the sequence $\{\mbf{k}_t\}_{t = 1}^T$. Each $\mbf{k}_t$ is a low-pass filter that should allow frequencies up till a threshold $\tau_t$ and then strongly attenuate any frequencies above $\tau_t$. In order to encode that the initial steps of the reverse process should only use low frequency components of $\mby$, we require that $\tau_T$ should start off small and increase as $t$ decreases. We denote the sequence of $\tau_t$ as our frequency curriculum. As datasets can have varying frequency characteristics, our method gives flexibility to set the frequency curriculum in a data-dependent way, as explained in \Cref{sec:experiments}. 

\noindent \textbf{Efficient Implementation.} We efficiently compute $\phi_t(\cdot)$ in the Fourier domain. Specifically, $\mbC_t$ can be diagonalized as $\mbC_t = \mbf{F^{-1} M F}$, where $\mbf{F}$ denotes the Discrete Fourier Transform matrix \cite{treiber2013optimization}. Using this property along with the fact that $\mbf{F}$ is an orthogonal matrix with determinant $1$, we have that $p(\phi_t(\mby) \mid \tweedie) = p(\mbf{F} \phi_t(\mby) \mid \tweedie) = p(\mbf{M F y} \mid \tweedie)$. Thus, we bypass storing the large circulant matrix and compute $\phi_t(\mby)$ and $\phi_t(\mc{A}(\tweedie))$ by first taking the Fast Fourier Transform of the input and applying a binary mask in the frequency domain, which adds minimal computational overhead (see Appendix \ref{sec:app:experimental_details}). %Further, to maintain computational efficiency, we approximate $\eta_t = \kappa / \norm{\phi_t(\mby) - \phi_t(\mc{A}(\tweedie))}{2}$ for a constant $\kappa$, similar to \cite{chung2022diffusion}.

\begin{algorithm} 
\caption{Frequency Guided Posterior Sampling} \label{alg:ours}
\begin{algorithmic}[1]
\Require $\mby$, $\mc{A}$, $\{\phi_t(\cdot), \mbf{S}_t, \alpha_t, \tilde{\sigma}_t \}_{t = 1}^T$
\State $\mbx_T \sim \mc{N}(\mbf{0}, \mbf{I})$ 
\For{$t \gets T$ to $1$}
    \State $\bm{\hat{s}} \gets s_\theta(\mbx_t, t)$ \Comment{Predict $\score{p(\mbx_t)}$}
    \State $\bar{\alpha_t} \gets \prod_{j = 1}^i \alpha_j$
    \State $\tweedie \gets \frac{1}{\sqrt{\bar{\alpha}_t}} (\mbx_t + (1 - \bar{\alpha}_t)\bm{\hat{s}})$ \Comment{Tweedie estimate}
    \State $\mbf{z} \sim \mc{N}(\mbf{0}, \mbf{I})$
    \State $\mbx_{t - 1}' \gets \frac{\sqrt{\alpha_t}(1 - \bar{\alpha}_{t - 1})}{1 - \bar{\alpha}_t}\mbx_t + \frac{\sqrt{\bar{\alpha}_{t - 1}}(1 - \alpha_t)}{1 - \bar{\alpha}_t}\tweedie + \tilde{\sigma}_t \mbf{z}$ \Comment{DDIM Sampling}
    \State $\mbx_{t - 1} \gets \mbx_{t - 1}' - \mbf{S}_t \grad_{\mbx_t} \norm{\phi_t(\mby) - \phi_t(\mc{A}(\tweedie))}{2}^2$ \Comment{FGPS Approximation of $\score{p(\mby \mid \mbx_t)}$}
\EndFor
\end{algorithmic}
\end{algorithm}

\section{Theoretical Analysis} \label{sec:limitations}

In this section, we utilize distributional assumptions to provide a precise theoretical analysis of the approximation errors induced by traditional posterior sampling methods as well as our proposed method.

%We hypothesize this is due to a large approximation gap between the true conditional score and its estimate. Next, we provide theoretical evidence to support this hypothesis. 

\subsection{Mind the (Approximation) Gap} \label{sec:theory}

Since we rely on a spectral characterization of the reverse diffusion process, our analysis requires two key properties. 
\begin{enumerate}
    \item We need to analyze the frequency characteristics of the data as well as the interplay between the frequency characteristics of the data and the forward operator. 
    \item To quantify the approximation gap, we need to compute the true conditional score. 
\end{enumerate} 
For the first property, we look at the power spectral density of the data, as defined in \Cref{eq:psd}, which tells us how much the signal's power is distributed across different frequencies. For the second property, it is natural to suppose the clean data $\mbx_0$ comes from a multivariate Gaussian distribution. Although these two properties seem conflicting, the following lemma establishes an equivalence between the data power spectral density and the data covariance.
\begin{lemma} \label{lem:wiener-khinchin}
    Let $f_k = \frac{k}{n}$ for $k = 0, \dots, n - 1$ denote the DFT sample frequencies for a signal of length $n$. Then, there exists a covariance matrix $\bm{\Sigma}_f$ such that for $\mbx \sim \mc{N}(\mbf{0}, \bm{\Sigma}_f)$, the following two properties hold. First, the signal $\mbx$ follows a power law in the frequency domain with parameters $c, \beta > 0$, i.e., it has a power spectral density $S(f_k) = c |f_k|^{-\beta}$ for the non-zero DFT sample frequencies $f_k$. Second, the eigenvalues of $\bm{\Sigma}_f$ are precisely $c |f_k|^{-\beta}$.
\end{lemma}
\noindent This is a special case of a classical result known as the Wiener-Khinchin Theorem \cite{champeney1987handbook_ch11}. Note that this does not imply that all random vectors that follow a power law in the frequency domain are Gaussian. For example, natural image data tends to follow a radially averaged power law in the frequency domain, but is far from Gaussian \cite{van1996modelling}. Nevertheless, this lemma gives a distributional assumption that mimics the spectral properties of natural image data, which we will show already elucidates failure cases of existing methods. Next, we use these assumptions to exactly quantify the approximation gap of our method. 

\begin{theorem} \label{thm:approx_gap}
Suppose $\mbx_0$ is drawn from $\mc{N}(\mbf{0}, \bm{\Sigma})$ and we are given linear measurements $\mby = \mc{A}(\mbx_0)+ \mbf{z}$, where $\mc{A}(\mbx_0) = \mbA \mbx_0$ and $\mbf{z} \sim \mc{N}(0, \sigma_y^2 I)$. Suppose that the intermediate value $\mbx_t$ of the continuous-time reverse diffusion from \Cref{eq:reverse_cond_sde} can be written as $\mbx_t = \sqrt{\bar{\alpha}_t} \mbx_0 + \sqrt{1 - \bar{\alpha}_t} \epsilon$ where $\epsilon \sim \mc{N}(\mbf{0}, \mbf{I})$.  Then, we have that the true noisy likelihood score $\score{p(\mby \mid \mbx_t)}$ is
     \begin{equation} \label{eq:true_cond_score}
         \!\!\!
        (\mbA \bm{\Gamma}_t)^T \!(\mbA \bm{\Sigma}_{0 \mid t} \mbA^T \!+ \sigma_y^2\bm{I})^{-1} ( \mby - \mbA \tweedie ),\!\!
     \end{equation}
where $\bm{\Gamma}_t = \sqrt{\bar{\alpha}_t} \bm{\Sigma} (\bar{\alpha}_t \bm{\Sigma} + (1 - \bar{\alpha}_t) \bm{I})^{-1}$, $\bm{\mu}_{0 \mid t} = \E[\mbx_0 \mid \mbx_t] = \bm{\Gamma}_t \mbx_t$ and $\bm{\Sigma}_{0 \mid t} = \Cov [\mbx_0 \mid \mbx_t] = \bm{\Sigma} - \sqrt{\bar{\alpha}_t} \bm{\Gamma}_t \bm{\Sigma}$. Moreover, the FGPS approximation $\score{p(\mbf{C}_t \mby \mid \tweedie)}$ can be analytically calculated as
     \begin{equation} \label{eq:our_cond_score}
         \left(\mbf{C}_t \mbA \bm{\Gamma}_t \right)^T (\sigma_y^{2} \mbf{C}_t \mbf{C}_t^T)^{-1} (\mbf{C}_t \mby - \mbf{C}_t \mbA \tweedie) .
     \end{equation}
\end{theorem}
\noindent The proof is in Appendix \ref{sec:app:proofs}. The difference of \Cref{eq:true_cond_score} and \Cref{eq:our_cond_score} gives a precise expression for the approximation gap of our method. The $\bar{\alpha}_t$ corresponds to the variance schedule of standard diffusion models such as DDPM \cite{ho2020denoising}. Our assumption that $\mbx_t = \sqrt{\bar{\alpha}_t} \mbx_0 + \sqrt{1 - \bar{\alpha}_t} \epsilon$ isolates the effects of the approximation gap as the generated sample is exactly $\mbx_0$. Our framework also allows us to analytically study the approximation gap for other inverse problem solvers. Below, we give one example for the popular Diffusion Posterior Sampling (DPS) solver.

\begin{corollary} \label{cor:dps_approx}
    Under the same assumptions and notation as in \Cref{thm:approx_gap}, we have that the DPS approximation $\score{p(\mby \mid \tweedie)}$ can be analytically calculated as
    \begin{equation} \label{eq:dps_cond_score}
         \!\!\!
        \sigma_y^{-2} (\mbA \bm{\Gamma}_t)^T (\mby - \mbA \tweedie ).\!\!
     \end{equation}
\end{corollary}

\noindent To the best of our knowledge, this is the first time any distributional assumptions have been used to characterize the exact approximation gap of any inverse problem solver. While \cite{chung2022diffusion} and \cite{yang2024guidance} provide upper and lower bounds on this gap, these bounds do not exploit the structure of the data distribution and as such, can be loose. We formally compare the approximation gap of FGPS and DPS in Appendix \ref{sec:app:theory}.

%\subsection{Empirical Study} \label{sec:theory_exp}

\subsection{Approximation Gap on Synthetic Data}

\begin{figure}[htbp]
    \centering
\includegraphics[width=0.42\textwidth]{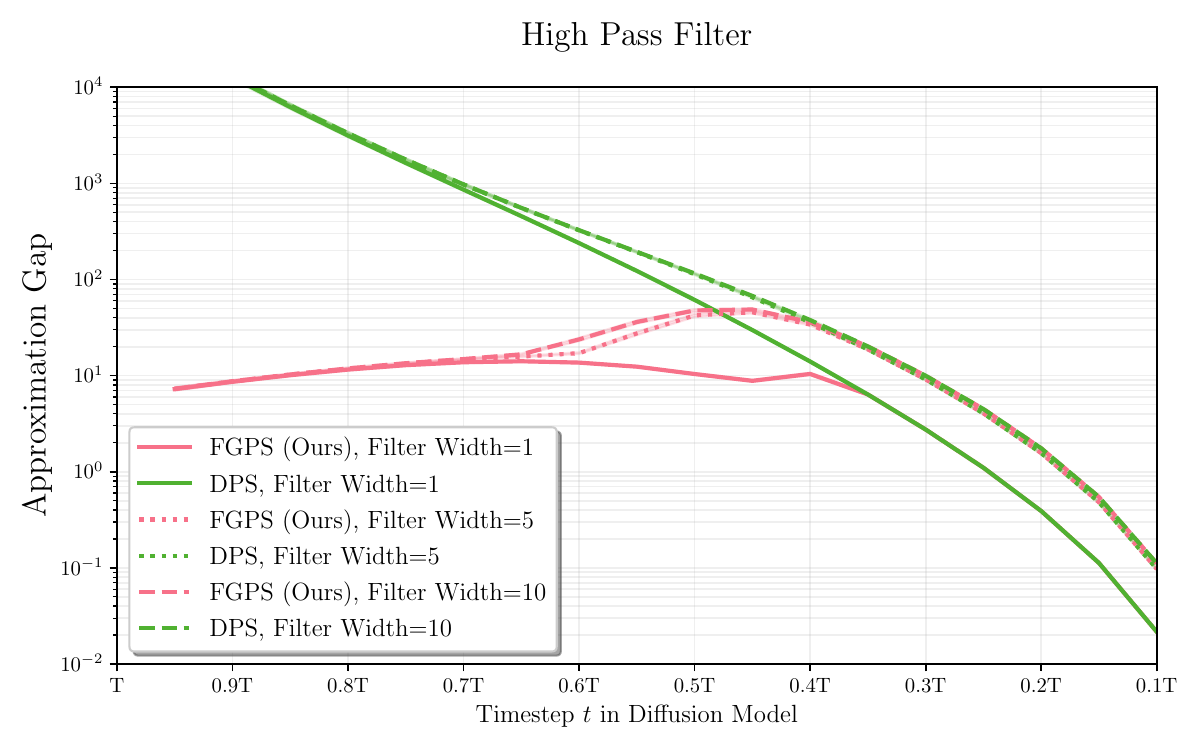}
    \caption{We show the exact approximation gap for FGPS and DPS across diffusion timesteps when the data follows a power law in the frequency domain and the forward operator is a a high-pass filter (Dirac kernel minus a Gaussian kernel) of varying width. The width denotes the $\sigma$ of the Gaussian kernel. Note the $y$-axis is in log-scale.} 
    \label{fig:approx_gap}
\end{figure}

 In this subsection, we probe the approximation gap of our method and DPS on synthetic data, and we give examples of forward operators where FGPS results in a significantly less approximation gap than DPS. In order to construct these examples, we note that \Cref{thm:approx_gap} depends crucially on how the forward operator $\mbA$ interacts with the data covariance $\bm{\Sigma}$, particularly through the term $\mbA \bm{\Sigma}_{0 \mid t} \mbA^T$. Intuitively, if the data covariance is chosen to be $\bm{\Sigma}_f$ from Lemma \ref{lem:wiener-khinchin}, the frequency characteristics of the data and the forward operator would be misaligned when the forward operator is a high-pass filter, so the measurement contains primarily high-frequency components. Some natural imaging systems fall into this regime, e.g., phase contrast microscopy where intensity measurements are high-pass filtered versions of refractive index.
 
 To verify the above intuition, we begin by simulating 1-D signals $\mbx_0 \sim \mc{N}(\mbf{0}, \bm{\Sigma}_f)$ that follow a power law in the frequency domain, which uses the construction of $\bm{\Sigma}_f$ given in Lemma \ref{lem:wiener-khinchin}. Then, the corresponding measurements for these signals are $\mby = \mbA \mbx_0 + \mbf{z}$ with $\mbf{z} \sim \mc{N}(\mbf{0}, \mbf{I})$ , where $\mbA$ is the convolution matrix corresponding to convolution with a Dirac kernel minus a Gaussian kernel of standard deviation $\sigma$. This kernel approximates a high-pass filter.  
 
 Figure \ref{fig:approx_gap} shows the approximation gap of DPS and FGPS as a function of the diffusion timesteps (where higher timesteps correspond to a lower value of $\bar{\alpha}_t$) for the variance schedule used in the DDPM model for $T = 1000$ \cite{ho2020denoising}. In accordance with our intuition, we see that the DPS approximation worsens significantly as we increase $t$, likely because higher frequency noise components in $\mbx_t$ are amplified by the high-pass filter. In contrast,  we find our method has an approximation gap that is several orders of magnitude smaller than DPS in the early stages of the reverse process. As the early stages of the reverse process are when the generated image develops coarse features, we argue it is crucial that the approximation gap is small in these phases. Refer to Appendix \ref{sec:app:experimental_details} for the full experimental details.

\section{Experiments}
\label{sec:experiments}

Although our theoretical results were restricted to cases where FGPS provably improved over existing methods, we now turn to empirical validation of FGPS on several real-world datasets and image restoration tasks. 

\begin{figure*}[htbp]
    \centering
    \includegraphics[width=0.73\textwidth]{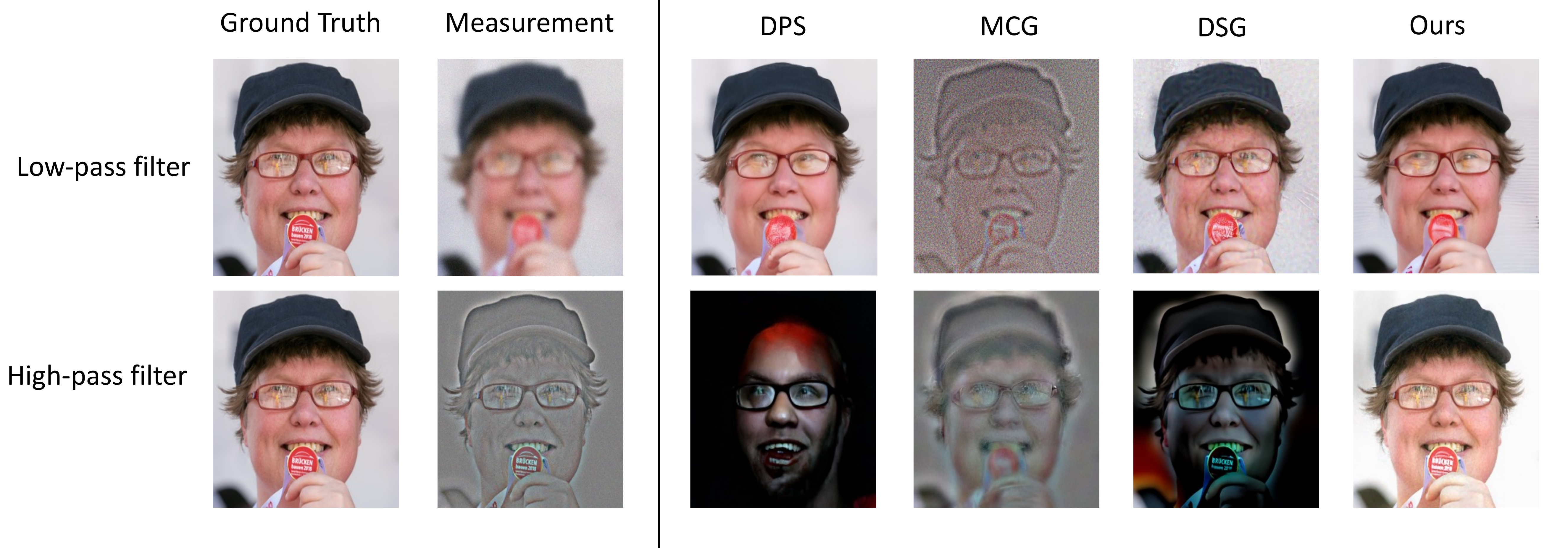}
    \caption{Even for linear inverse problems,  SOTA diffusion-based methods (DPS, MCG, DSG) can struggle to solve inverse problems when the forward operator is convolution with a high-pass filter. In contrast, our method is still able to achieve a high quality reconstruction.}
    \label{fig:highpass_limitations}
    \vspace{-1mm}
\end{figure*}

\begin{table*}[ht] 
\centering
%\small  % Changed from default size to small
\footnotesize
\setlength{\tabcolsep}{2.5pt}  % Reduced horizontal padding between columns
\begin{tabular}{l|cccc|cccc|cccc}
\toprule
& \multicolumn{4}{c|}{Gaussian Blur} & \multicolumn{4}{c|}{Motion Blur} & \multicolumn{4}{c}{High-Pass Filter} \\
Method & FID$\downarrow$ & LPIPS$\downarrow$ & PSNR$\uparrow$ & SSIM$\uparrow$ & FID$\downarrow$ & LPIPS$\downarrow$ & PSNR$\uparrow$ & SSIM$\uparrow$ & FID$\downarrow$ & LPIPS$\downarrow$ & PSNR$\uparrow$ & SSIM$\uparrow$ \\
\midrule
\multicolumn{12}{c}{\textbf{FFHQ Dataset}} \\
\midrule
Score-SDE/ILVR \cite{choi2021ilvr, song2020score} & 58.79 & 0.266 & 23.34 & 0.648 & 69.10 & 0.387 & 18.59 & 0.457 & 85.35 & 0.356 & 11.86 & 0.600 \\
MCG \cite{chung2022improving} & 310.33 & 1.035 & 11.64 & 0.143 & 296.93 & 0.771 & 11.54 & 0.127 & 151.77 & 0.581 & \underline{12.48} & 0.466 \\
DPS \cite{chung2022diffusion} & \textbf{21.69} & \textbf{0.116} &  \textbf{26.01} & \textbf{0.743} & \underline{27.15} & \underline{0.153} &  24.17 & \underline{0.688} & 121.57 & 0.694 & 6.17& 0.222 \\
DSG \cite{yang2024guidance} & 44.63 & 0.294 & 22.75 & 0.633 & 32.74 & 0.257 & \textbf{25.66} & 0.687 & \underline{17.81} & \underline{0.246} & 12.12 & \underline{0.695} \\
FGPS (Ours) & \underline{21.89} & \underline{0.139} & \underline{25.44} & \underline{0.727}  & \textbf{20.48} & \textbf{0.127} & \underline{25.48} & \textbf{0.723}  & \textbf{10.02} & \textbf{0.124} & \textbf{18.79} & \textbf{0.739} \\
\midrule
\multicolumn{12}{c}{\textbf{ImageNet Dataset}} \\
\midrule
Score-SDE/ILVR \cite{choi2021ilvr, song2020score} & 104.81 & 0.676 & 16.27 & 0.318 & 100.54 & 0.665 & 16.62 & 0.290 & 226.51 & 0.784 & 10.69 & 0.345 \\
MCG \cite{chung2022improving} & 320.63 & 1.005 & 12.01 & 0.121 & 306.43 & 0.963 & 11.80 & 0.076 & 295.10 & 1.000 & \underline{11.62} & 0.278 \\
DPS \cite{chung2022diffusion} & 98.02 & 0.517 & 17.35 & 0.364 & \underline{77.74} & \underline{0.407} & \underline{20.90} & \underline{0.551} & 111.73 & 0.569 & 10.66 & \underline{0.396} \\
DSG \cite{yang2024guidance} & \underline{87.51} & \underline{0.375} & \underline{20.75} & \underline{0.524} & 107.25 & 0.453 & 19.58 & 0.488 & \underline{110.85} & \underline{0.546} & 8.425 & 0.382 \\
FGPS (Ours) & \textbf{56.46} & \textbf{0.294} & \textbf{21.70} & \textbf{0.574} & \textbf{49.25} & \textbf{0.267} & \textbf{22.01} & \textbf{0.601} & \textbf{24.44} & \textbf{0.192} & \textbf{15.96} & \textbf{0.686}  \\
\bottomrule
\end{tabular}
\caption{Quantitative comparison of different methods on FFHQ and ImageNet validation datasets. Bolded entries indicate best performance and underlined entries indicate second best performance.}
\label{tab:linear_inv}
\end{table*} 

\begin{figure*}[ht]
    \centering
    \includegraphics[width=\textwidth]{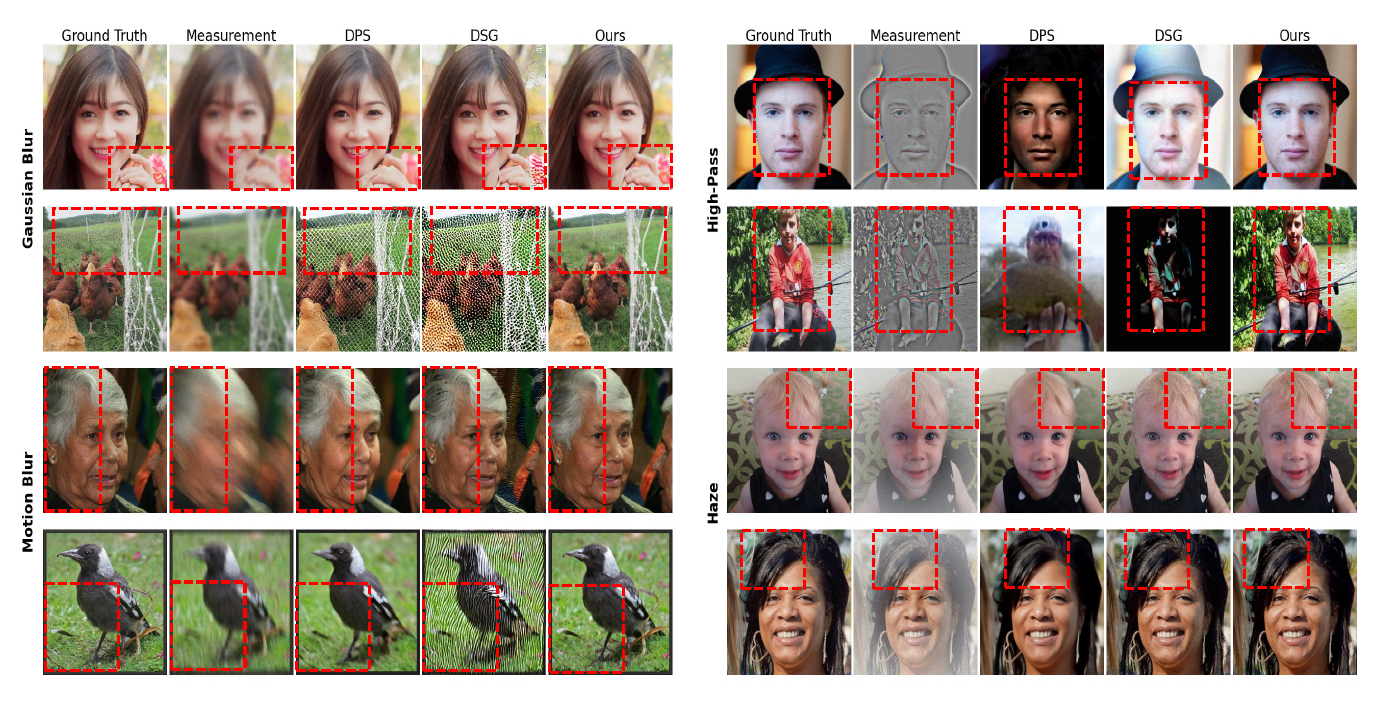}
    \caption{Qualitative results on FFHQ and ImageNet datasets. The dotted red boxes highlight areas of the image where our method results in higher-quality reconstructions than DPS and DSG. Zoomed in versions can be found in the Appendix.}
    \label{fig:main_results}
\end{figure*}

\subsection{Experimental Setup}

\noindent \textbf{Datasets.} We evaluate our method on two datasets: the Flickr-Faces-HQ (FFHQ 256x256) dataset of high-resolution faces \cite{Karras2019ASG}, and the ImageNet-1000 dataset of diverse objects \cite{deng2009imagenet}. We utilize pretrained diffusion models from \cite{chung2022diffusion} and \cite{dhariwal2021diffusion} for the FFHQ and ImageNet datasets respectively, using a class-conditional model for ImageNet.

\noindent \textbf{Tasks.} We consider four challenging image restoration tasks for our experimental evaluation. The first three are linear inverse problems: Gaussian deblurring, motion deblurring, and deconvolution with a high-pass filter. The last is a nonlinear inverse problem: image dehazing. The details of the four forward operators we consider can be found in Appendix \ref{sec:app:experimental_details}. All measurements have Gaussian measurement noise added with standard deviation $0.05$. 

\noindent \textbf{Baselines and Metrics.} We compare against several competitive unsupervised methods. The first two baselines are popular unsupervised diffusion-based solvers: Diffusion Posterior Sampling (DPS) \cite{chung2022diffusion}, a method that approximates the conditional score with a Dirac delta around the posterior mean, and Score-SDE/ILVR \cite{choi2021ilvr, song2020score}, a method that uses a sequence of noisy measurements computed using the diffusion model variance schedule. Then, we compare to two state-of-the-art methods focused on mitigating the approximation gap problem: Manifold Constrained Gradient (MCG) \cite{chung2022improving} and Diffusion with Spherical Guidance (DSG) \cite{yang2024guidance}. For image dehazing, we further compare to two problem-specific baselines: DoubleDIP \cite{gandelsman2019double} and AOD-Net \cite{li2017aod}, a supervised method trained to dehaze images from pairs of clean-hazy images. We use the same pretrained score function for all diffusion-based methods. We evaluate our method on two perceptual metrics, FID and LPIPS distance, and two distortion-based metrics, PSNR and SSIM.

\noindent \textbf{Frequency Curriculum.} For both FFHQ and ImageNet across all tasks, in our method we set $\tau_1$ and $\tau_T$ to be $\frac{75*100}{256}\%$ and $\frac{10*100}{256}\%$ of the overall frequency range respectively. We set different frequency curricula for the two datasets as described in \Cref{sec:method}. We find the exponential schedule works well for face datasets with many high-frequency details, while the uniform schedule is better for diverse, multiclass datasets like ImageNet. We emphasize that $\mbf{k}_t$ does not require extensive tuning, as a simple schedule already beats baseline methods (see Appendix \ref{sec:app:further_exp}). 

%Specifically, our experiments (shown in Appendix \ref{sec:app:further_exp}) reveal that on deblurring tasks on FFHQ, our method performs better with an exponentially increasing schedule between $\tau_T$ and $\tau_1$, while on ImageNet, our method achieves better results with a linear schedule. %This is intuitive since for FFHQ, the faces will have finer facial features that need to be refined for many steps in the later stages of the reverse diffusion process.

\noindent \textbf{Step Size Schedule.} In practice, we take $\mbf{S}_t$ to be a hyperparameter as opposed to being fixed as $(\sigma_y^2 \mbf{C}_t \mbf{C}_t^T)^{-1}$ in order to further control the effect of the approximation error. We adopt the choice that DPS makes where the step size is $\frac{\kappa}{\norm{\phi_t(\mby) - \phi(\mc{A}(\tweedie))}{2}} \mbf{I}$ for a hyperparameter $\kappa$. Using the intuition from \Cref{fig:approx_gap} that the approximation gap is small in the initial steps of the reverse process, we set $\kappa$ large initially and decay it following a cosine curve. %We find this step size schedule results in a significant improvement in reconstruction quality on a held-out validation set. 

\subsection{Results}

\noindent \textbf{Empirical Validation of Theoretical Results.} In Figure \ref{fig:highpass_limitations}, we consider two forward operators on an example from the FFHQ dataset: a) a standard Gaussian blurring low-pass filter, and b) the high-pass filter we studied in our theoretical analysis. As expected, prior approaches such as Diffusion Posterior Sampling 
\cite{chung2022diffusion} and variants \cite{chung2022improving, yang2024guidance} perform well at Gaussian deblurring. However, when the forward operator is a high-pass filter, remarkably these approaches all fail to produce coherent outputs, which aligns with our theoretical results about the large approximation gap of existing methods. FGPS on the other hand is still able to produce high-quality reconstructions.

\noindent \textbf{Results on Linear Inverse Problems.} \Cref{tab:linear_inv} shows our main results for linear inverse problems. On the FFHQ dataset, for simple tasks such as Gaussian deblurring, FGPS gives comparable performance to baselines, which already have strong reconstruction performance. However, on more complex image restoration tasks such as motion deblurring, our method outperforms the baselines in almost all metrics. Notably, even though Score-SDE/ILVR uses a sequence of noisy measurements, this sequence is not adapted to the frequency characteristics of the data, which results in lower quality reconstructions as compared to our method (see Appendix \ref{sec:app:ilvr} for further discussion). In \Cref{fig:main_results}, we see the qualitative differences between the methods. For high-pass filter operators, existing methods generate visually implausible images, indicating that the approximation errors push the denoised outputs away from the natural image manifold. In contrast, our method is able to accurately reconstruct the underlying ground truth image. For blurry images, we see that DPS tends to produce smoothed images, losing high-frequency details. While DSG is able to improve upon this and generate images with fine-grained characteristics of the underlying image, it tends to also generate artifacts. In contrast, due to our frequency curriculum, our method is able to generate images with high perceptual quality. We give further qualitative comparisons to baselines in Appendix \ref{sec:app:further_exp}.

\noindent \textbf{Role of Dataset.} For the ImageNet dataset, we observe that FGPS significantly outperforms all baselines, especially in perceptual quality. We conjecture that this is because on more complex datasets with varying frequency characteristics, the explicit frequency schedule for the measurement is essential to ensure the reverse process iterates do not deviate far from the intermediate noisy data manifolds learned by the diffusion models.

\noindent \textbf{Results on Image Dehazing.} Table \ref{tab:haze_comparison} shows our main results for the nonlinear image dehazing task on the FFHQ dataset. A key strength of our method is the ability to handle nonlinear inverse problems, unlike methods such as \cite{kawar2022denoising, song2020score, choi2021ilvr}. In \Cref{fig:main_results}, we observe that our method successfully reconstructs finer details such as background details and facial characteristics similar to what we observed for linear inverse problems. DSG produces slightly grainy reconstructions, but performs much better on the dehazing task than the other tasks. Due to the powerful generative capacity of diffusion models, our method also performs better than problem-specific supervised methods such as AOD-Net \cite{li2017aod}, a phenomenon also observed in \cite{song2023pseudoinverse}. We provide more qualitative examples in Appendix \ref{sec:app:further_exp}.

\begin{table}[H]
\centering
\begin{tabular}{l|c|c|c|c}
\toprule
%& \multicolumn{2}{c}{Perceptual} & \multicolumn{2}{c}{Distortion}  \\
%\cline{2-3} \cline{4-5} 
Method & FID$\downarrow$ & LPIPS$\downarrow$ & PSNR$\uparrow$ & SSIM$\uparrow$ \\
\hline
DoubleDIP \cite{gandelsman2019double} & 50.51 & 0.500 & 13.80 & 0.483 \\
AOD-Net \cite{li2017aod} & 33.35 & 0.325 & 17.42 & 0.532 \\
DPS \cite{chung2022diffusion} & 24.24 & 0.137 & 25.88 & 0.806 \\
DSG \cite{yang2024guidance} & \underline{10.55} & \underline{0.134} & \underline{26.07} & \underline{0.795} \\
FGPS (Ours) & \textbf{9.37} & \textbf{0.066} & \textbf{26.88} & \textbf{0.840} \\
\hline
\end{tabular}
\caption{Comparison of different haze removal methods. Bolded entries indicate best performance and underlined entries indicate second best performance.}
\label{tab:haze_comparison}
\vspace{-3mm}
\end{table}

\begin{figure}[htbp]
    \vspace{-2mm}
    \centering
\includegraphics[width=0.27\textwidth]{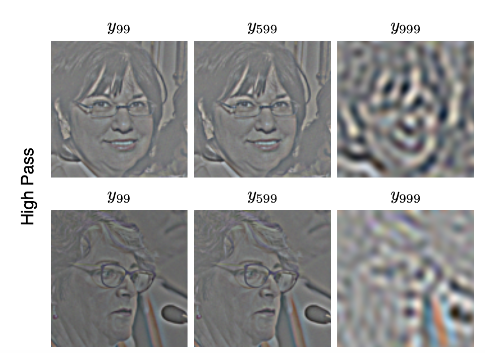}
    \caption{Visualization of Transformed Measurements $y_t$ at different timesteps $t$, demonstrating our coarse-to-fine strategy.}
    \label{fig:cond_meas}
    \vspace{-5mm}
\end{figure}

\noindent \textbf{Visualizing Measurement Sequence.} In \Cref{fig:cond_meas}, we show the transformed measurements at three timesteps in the reverse diffusion process when running FGPS. In the early stages of the reverse process, the measurement retains only a coarse outline of the measurement. This coarse-to-fine strategy is essential to ensuring $\mbx_t$ remains on the data manifold learned by the diffusion model.

\noindent \textbf{Effect of Frequency Curriculum.} A key ingredient of our method is the frequency curriculum with a time-varying low pass filter $\mbf{k}_t$. To understand the effect of this curriculum, we compare taking a time-dependent $\mbf{k}_t$ vs. a fixed $\mbf{k}$ that zeroes out frequencies above $\tau_1$ for all $t$. In \Cref{fig:ablation_kt}, we see that for the high-pass filter operator, the generated reconstructions are still visually implausible images for a fixed $\mbf{k}$, indicating the frequency curriculum is crucial to our method. Our method only requires $\mbf{k}_1$ to be a low frequency cut-off, otherwise the reconstructions look like \Cref{fig:ablation_kt}. This corroborates our theoretical results, which show that the approximation gap early on in the reverse process greatly influences reconstruction quality. Finally, our method does not qualitatively change with higher cutoffs for $\mbf{k}_T$. We examine this in detail in Appendix \ref{sec:app:further_exp} as well as conduct further ablation studies. 

\begin{figure}[htbp]
    \centering
\includegraphics[width=0.33\textwidth]{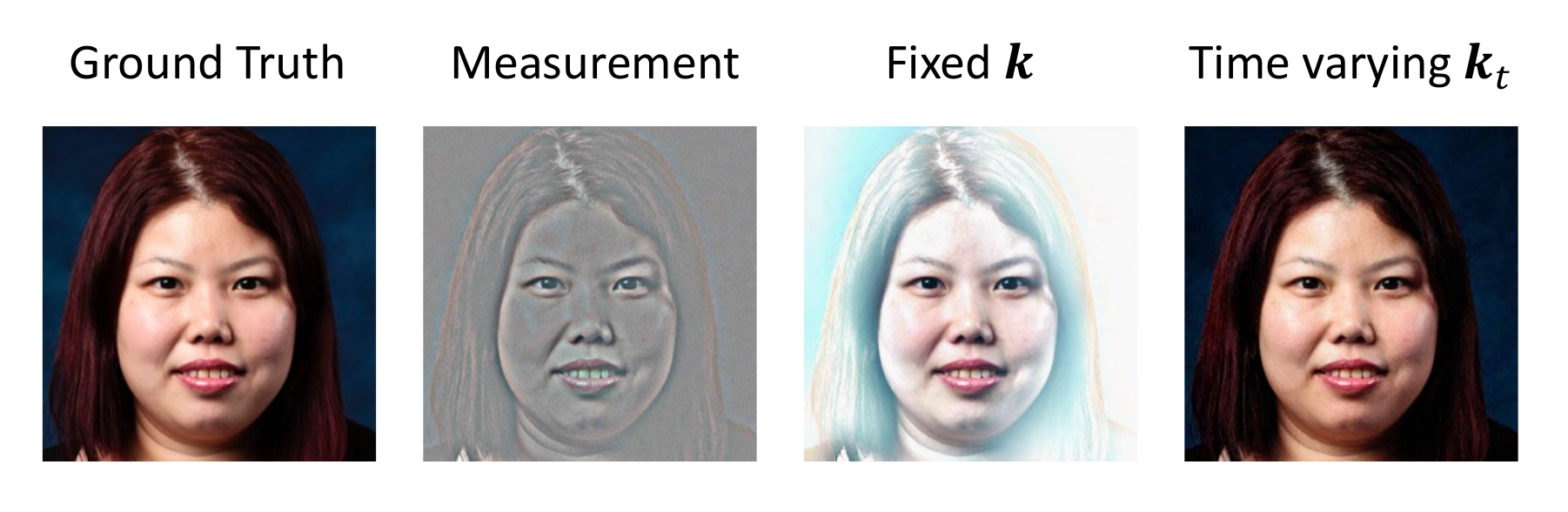}
    \caption{Qualitative Ablation Study: fixed low-pass filter vs. time-dependent low pass filter.}
    \label{fig:ablation_kt}
    \vspace{-1mm}
\end{figure}

\section{Conclusion}

In this paper, we introduce Frequency Guided Posterior Sampling (FGPS), a method that progressively incorporates frequency information in a data-dependent manner. We provide a theoretical analysis characterizing the approximation error of FGPS and existing diffusion-based restoration methods, demonstrating how these errors can become severe when forward operators interact with high-frequency components of images. Our experiments on challenging image restoration tasks, such as motion deblurring and image dehazing, demonstrate that FGPS significantly outperforms existing methods. This work opens up new directions for analyzing diffusion-based inverse problem solvers by carefully considering the frequency characteristics of both the data and the degradation process.

\section*{Acknowledgements}

D.T. was supported by a gift from AWS to Penn Engineering's ASSET Center for Trustworthy AI. R.V. thanks the support of Penn startup funds and the Research Collaboration on the Mathematical and Scientific Foundations of Deep Learning under grants NSF 2031985 and Simons 814201.

{
    \small
    \bibliographystyle{ieeenat_fullname}
    \bibliography{main}

\begin{thebibliography}{42}
\providecommand{\natexlab}[1]{#1}
\providecommand{\url}[1]{\texttt{#1}}
\expandafter\ifx\csname urlstyle\endcsname\relax
  \providecommand{\doi}[1]{doi: #1}\else
  \providecommand{\doi}{doi: \begingroup \urlstyle{rm}\Url}\fi

\bibitem[Anderson(1982)]{anderson1982reverse}
Brian~DO Anderson.
\newblock Reverse-time diffusion equation models.
\newblock \emph{Stochastic Processes and their Applications}, 12\penalty0 (3):\penalty0 313--326, 1982.

\bibitem[Bishop and Nasrabadi(2006)]{bishop2006pattern}
Christopher~M Bishop and Nasser~M Nasrabadi.
\newblock \emph{Pattern recognition and machine learning}.
\newblock Springer, 2006.

\bibitem[Bora et~al.(2017)Bora, Jalal, Price, and Dimakis]{bora2017compressed}
Ashish Bora, Ajil Jalal, Eric Price, and Alexandros~G Dimakis.
\newblock Compressed sensing using generative models.
\newblock In \emph{International conference on machine learning}, pages 537--546. PMLR, 2017.

\bibitem[Candes and Recht(2012)]{candes2012exact}
Emmanuel Candes and Benjamin Recht.
\newblock Exact matrix completion via convex optimization.
\newblock \emph{Communications of the ACM}, 55\penalty0 (6):\penalty0 111--119, 2012.

\bibitem[Candes et~al.(2006)Candes, Romberg, and Tao]{candes2006stable}
Emmanuel~J Candes, Justin~K Romberg, and Terence Tao.
\newblock Stable signal recovery from incomplete and inaccurate measurements.
\newblock \emph{Communications on Pure and Applied Mathematics: A Journal Issued by the Courant Institute of Mathematical Sciences}, 59\penalty0 (8):\penalty0 1207--1223, 2006.

\bibitem[Cand{\`e}s et~al.(2011)Cand{\`e}s, Li, Ma, and Wright]{candes2011robust}
Emmanuel~J Cand{\`e}s, Xiaodong Li, Yi Ma, and John Wright.
\newblock Robust principal component analysis?
\newblock \emph{Journal of the ACM (JACM)}, 58\penalty0 (3):\penalty0 1--37, 2011.

\bibitem[Champeney(1987)]{champeney1987handbook_ch11}
David~C Champeney.
\newblock \emph{A handbook of Fourier theorems}, chapter~11.
\newblock Cambridge University Press, 1987.

\bibitem[Choi et~al.(2021)Choi, Kim, Jeong, Gwon, and Yoon]{choi2021ilvr}
Jooyoung Choi, Sungwon Kim, Yonghyun Jeong, Youngjune Gwon, and Sungroh Yoon.
\newblock Ilvr: Conditioning method for denoising diffusion probabilistic models.
\newblock \emph{arXiv preprint arXiv:2108.02938}, 2021.

\bibitem[Chung et~al.(2022{\natexlab{a}})Chung, Kim, Mccann, Klasky, and Ye]{chung2022diffusion}
Hyungjin Chung, Jeongsol Kim, Michael~T Mccann, Marc~L Klasky, and Jong~Chul Ye.
\newblock Diffusion posterior sampling for general noisy inverse problems.
\newblock \emph{arXiv preprint arXiv:2209.14687}, 2022{\natexlab{a}}.

\bibitem[Chung et~al.(2022{\natexlab{b}})Chung, Sim, Ryu, and Ye]{chung2022improving}
Hyungjin Chung, Byeongsu Sim, Dohoon Ryu, and Jong~Chul Ye.
\newblock Improving diffusion models for inverse problems using manifold constraints.
\newblock \emph{Advances in Neural Information Processing Systems}, 35:\penalty0 25683--25696, 2022{\natexlab{b}}.

\bibitem[Daniels et~al.(2021)Daniels, Maunu, and Hand]{daniels2021score}
Max Daniels, Tyler Maunu, and Paul Hand.
\newblock Score-based generative neural networks for large-scale optimal transport.
\newblock \emph{Advances in neural information processing systems}, 34:\penalty0 12955--12965, 2021.

\bibitem[Daras et~al.(2024)Daras, Chung, Lai, Mitsufuji, Ye, Milanfar, Dimakis, and Delbracio]{daras2024survey}
Giannis Daras, Hyungjin Chung, Chieh-Hsin Lai, Yuki Mitsufuji, Jong~Chul Ye, Peyman Milanfar, Alexandros~G Dimakis, and Mauricio Delbracio.
\newblock A survey on diffusion models for inverse problems.
\newblock \emph{arXiv preprint arXiv:2410.00083}, 2024.

\bibitem[Deng et~al.(2009)Deng, Dong, Socher, Li, Li, and Fei-Fei]{deng2009imagenet}
Jia Deng, Wei Dong, Richard Socher, Li-Jia Li, Kai Li, and Li Fei-Fei.
\newblock Imagenet: A large-scale hierarchical image database.
\newblock In \emph{2009 IEEE conference on computer vision and pattern recognition}, pages 248--255. Ieee, 2009.

\bibitem[Dhariwal and Nichol(2021)]{dhariwal2021diffusion}
Prafulla Dhariwal and Alexander Nichol.
\newblock Diffusion models beat gans on image synthesis.
\newblock \emph{Advances in neural information processing systems}, 34:\penalty0 8780--8794, 2021.

\bibitem[Dieleman(2024)]{dieleman2024spectral}
Sander Dieleman.
\newblock Diffusion is spectral autoregression, 2024.

\bibitem[Donoho et~al.(2005)Donoho, Elad, and Temlyakov]{donoho2005stable}
David~L Donoho, Michael Elad, and Vladimir~N Temlyakov.
\newblock Stable recovery of sparse overcomplete representations in the presence of noise.
\newblock \emph{IEEE Transactions on information theory}, 52\penalty0 (1):\penalty0 6--18, 2005.

\bibitem[Efron(2011)]{efron2011tweedie}
Bradley Efron.
\newblock Tweedie’s formula and selection bias.
\newblock \emph{Journal of the American Statistical Association}, 106\penalty0 (496):\penalty0 1602--1614, 2011.

\bibitem[Gandelsman et~al.(2019)Gandelsman, Shocher, and Irani]{gandelsman2019double}
Yosef Gandelsman, Assaf Shocher, and Michal Irani.
\newblock " double-dip": unsupervised image decomposition via coupled deep-image-priors.
\newblock In \emph{Proceedings of the IEEE/CVF conference on computer vision and pattern recognition}, pages 11026--11035, 2019.

\bibitem[Ho and Salimans(2022)]{ho2022classifier}
Jonathan Ho and Tim Salimans.
\newblock Classifier-free diffusion guidance.
\newblock \emph{arXiv preprint arXiv:2207.12598}, 2022.

\bibitem[Ho et~al.(2020)Ho, Jain, and Abbeel]{ho2020denoising}
Jonathan Ho, Ajay Jain, and Pieter Abbeel.
\newblock Denoising diffusion probabilistic models.
\newblock \emph{Advances in neural information processing systems}, 33:\penalty0 6840--6851, 2020.

\bibitem[Ho et~al.(2022)Ho, Salimans, Gritsenko, Chan, Norouzi, and Fleet]{ho2022video}
Jonathan Ho, Tim Salimans, Alexey Gritsenko, William Chan, Mohammad Norouzi, and David~J Fleet.
\newblock Video diffusion models.
\newblock \emph{Advances in Neural Information Processing Systems}, 35:\penalty0 8633--8646, 2022.

\bibitem[Karras et~al.(2019)Karras, Laine, and Aila]{Karras2019ASG}
Tero Karras, S. Laine, and Timo Aila.
\newblock A style-based generator architecture for generative adversarial networks.
\newblock \emph{2019 IEEE/CVF Conference on Computer Vision and Pattern Recognition (CVPR)}, pages 4396--4405, 2019.

\bibitem[Kawar et~al.(2021)Kawar, Vaksman, and Elad]{kawar2021snips}
Bahjat Kawar, Gregory Vaksman, and Michael Elad.
\newblock Snips: Solving noisy inverse problems stochastically.
\newblock \emph{Advances in Neural Information Processing Systems}, 34:\penalty0 21757--21769, 2021.

\bibitem[Kawar et~al.(2022)Kawar, Elad, Ermon, and Song]{kawar2022denoising}
Bahjat Kawar, Michael Elad, Stefano Ermon, and Jiaming Song.
\newblock Denoising diffusion restoration models.
\newblock \emph{Advances in Neural Information Processing Systems}, 35:\penalty0 23593--23606, 2022.

\bibitem[Kong et~al.(2020)Kong, Ping, Huang, Zhao, and Catanzaro]{kong2020diffwave}
Zhifeng Kong, Wei Ping, Jiaji Huang, Kexin Zhao, and Bryan Catanzaro.
\newblock Diffwave: A versatile diffusion model for audio synthesis.
\newblock \emph{arXiv preprint arXiv:2009.09761}, 2020.

\bibitem[Li et~al.(2017)Li, Peng, Wang, Xu, and Feng]{li2017aod}
Boyi Li, Xiulian Peng, Zhangyang Wang, Jizheng Xu, and Dan Feng.
\newblock Aod-net: All-in-one dehazing network.
\newblock In \emph{Proceedings of the IEEE international conference on computer vision}, pages 4770--4778, 2017.

\bibitem[Lugmayr et~al.(2022)Lugmayr, Danelljan, Romero, Yu, Timofte, and Van~Gool]{lugmayr2022repaint}
Andreas Lugmayr, Martin Danelljan, Andres Romero, Fisher Yu, Radu Timofte, and Luc Van~Gool.
\newblock Repaint: Inpainting using denoising diffusion probabilistic models.
\newblock In \emph{Proceedings of the IEEE/CVF conference on computer vision and pattern recognition}, pages 11461--11471, 2022.

\bibitem[Lv et~al.(2024)Lv, Zhang, Wang, Zheng, Zhong, Li, and Nie]{lv2024fourier}
Xiaoqian Lv, Shengping Zhang, Chenyang Wang, Yichen Zheng, Bineng Zhong, Chongyi Li, and Liqiang Nie.
\newblock Fourier priors-guided diffusion for zero-shot joint low-light enhancement and deblurring.
\newblock In \emph{Proceedings of the IEEE/CVF Conference on Computer Vision and Pattern Recognition}, pages 25378--25388, 2024.

\bibitem[Pan et~al.(2021)Pan, Zhan, Dai, Lin, Loy, and Luo]{pan2021exploiting}
Xingang Pan, Xiaohang Zhan, Bo Dai, Dahua Lin, Chen~Change Loy, and Ping Luo.
\newblock Exploiting deep generative prior for versatile image restoration and manipulation.
\newblock \emph{IEEE Transactions on Pattern Analysis and Machine Intelligence}, 44\penalty0 (11):\penalty0 7474--7489, 2021.

\bibitem[Rissanen et~al.(2022)Rissanen, Heinonen, and Solin]{rissanen2022generative}
Severi Rissanen, Markus Heinonen, and Arno Solin.
\newblock Generative modelling with inverse heat dissipation.
\newblock \emph{arXiv preprint arXiv:2206.13397}, 2022.

\bibitem[Rozet et~al.(2024)Rozet, Andry, Lanusse, and Louppe]{rozet2024learning}
Fran{\c{c}}ois Rozet, G{\'e}r{\^o}me Andry, Fran{\c{c}}ois Lanusse, and Gilles Louppe.
\newblock Learning diffusion priors from observations by expectation maximization.
\newblock \emph{arXiv preprint arXiv:2405.13712}, 2024.

\bibitem[Rudin et~al.(1992)Rudin, Osher, and Fatemi]{rudin1992nonlinear}
Leonid~I Rudin, Stanley Osher, and Emad Fatemi.
\newblock Nonlinear total variation based noise removal algorithms.
\newblock \emph{Physica D: nonlinear phenomena}, 60\penalty0 (1-4):\penalty0 259--268, 1992.

\bibitem[Saharia et~al.(2022)Saharia, Chan, Chang, Lee, Ho, Salimans, Fleet, and Norouzi]{saharia2022palette}
Chitwan Saharia, William Chan, Huiwen Chang, Chris Lee, Jonathan Ho, Tim Salimans, David Fleet, and Mohammad Norouzi.
\newblock Palette: Image-to-image diffusion models.
\newblock In \emph{ACM SIGGRAPH 2022 conference proceedings}, pages 1--10, 2022.

\bibitem[Schaaf and Hateren(1996)]{van1996modelling}
A.~van~der Schaaf and J.~H.~van Hateren.
\newblock Modelling the {Power} {Spectra} of {Natural} {Images}: {Statistics} and {Information}.
\newblock \emph{Vision Research}, 36\penalty0 (17):\penalty0 2759--2770, 1996.

\bibitem[Song et~al.(2023)Song, Vahdat, Mardani, and Kautz]{song2023pseudoinverse}
Jiaming Song, Arash Vahdat, Morteza Mardani, and Jan Kautz.
\newblock Pseudoinverse-guided diffusion models for inverse problems.
\newblock In \emph{International Conference on Learning Representations}, 2023.

\bibitem[Song and Ermon(2019)]{song2019generative}
Yang Song and Stefano Ermon.
\newblock Generative modeling by estimating gradients of the data distribution.
\newblock \emph{Advances in neural information processing systems}, 32, 2019.

\bibitem[Song et~al.(2020)Song, Sohl-Dickstein, Kingma, Kumar, Ermon, and Poole]{song2020score}
Yang Song, Jascha Sohl-Dickstein, Diederik~P Kingma, Abhishek Kumar, Stefano Ermon, and Ben Poole.
\newblock Score-based generative modeling through stochastic differential equations.
\newblock \emph{arXiv preprint arXiv:2011.13456}, 2020.

\bibitem[Stuart(2010)]{stuart2010inverse}
Andrew~M Stuart.
\newblock Inverse problems: a bayesian perspective.
\newblock \emph{Acta numerica}, 19:\penalty0 451--559, 2010.

\bibitem[Treiber(2013)]{treiber2013optimization}
Marco~Alexander Treiber.
\newblock \emph{Optimization for computer vision}.
\newblock Springer, 2013.

\bibitem[Vincent(2011)]{vincent2011connection}
Pascal Vincent.
\newblock A connection between score matching and denoising autoencoders.
\newblock \emph{Neural computation}, 23\penalty0 (7):\penalty0 1661--1674, 2011.

\bibitem[Wang et~al.(2022)Wang, Yu, and Zhang]{wang2022zero}
Yinhuai Wang, Jiwen Yu, and Jian Zhang.
\newblock Zero-shot image restoration using denoising diffusion null-space model.
\newblock \emph{arXiv preprint arXiv:2212.00490}, 2022.

\bibitem[Yang et~al.(2024)Yang, Ding, Cai, Yu, Wang, and Shi]{yang2024guidance}
Lingxiao Yang, Shutong Ding, Yifan Cai, Jingyi Yu, Jingya Wang, and Ye Shi.
\newblock Guidance with spherical gaussian constraint for conditional diffusion.
\newblock \emph{arXiv preprint arXiv:2402.03201}, 2024.

\end{thebibliography}
}

\clearpage
\setcounter{page}{1}
\maketitlesupplementary

\appendix

\section{Further Background} \label{sec:app:background}

\subsection{Diffusion Models}

Recall that unconditional diffusion models learn the score function $\score{p(\mbx)}$ through denoising score matching. Assuming we have learned a neural network that approximates this score well, the backbone of most state of the art diffusion-based inverse problem solvers is the estimation of $\tweedie = \E[\mbx_0 \mid \mbx_t]$ using the learned diffusion model. This is known as Tweedie's formula, which we state below.

\begin{lemma} (Tweedie's formula \cite{efron2011tweedie})
    Suppose $p(\mbx_t \mid \mbx_0) = \mc{N}(\sqrt{\alpha_t} \mbx_0, (1 - \alpha_t) \mbf{I})$. Then the posterior mean is
    \begin{equation}
        \E[\mbx_0 \mid \mbx_t] = \frac{1}{\sqrt{\alpha_t}} (\mbx_t + (1 - \alpha_t) \score{p(\mbx_t)})    
    \end{equation}
\end{lemma}

\begin{proof}
    We expand the score function as
    \begin{align}
        \score{p(\mbx_t)} &= \frac{\grad_{\mbx_t} p(\mbx_t)}{p(\mbx_t)} \\
        &= \frac{1}{p(\mbx_t)} \int \grad_{\mbx_t} p(\mbx_t \mid \mbx_0) p(\mbx_0) \ d\mbx_0.
    \end{align}
    We can rewrite $\grad_{\mbx_t} p(\mbx_t \mid \mbx_0)$ as $p(\mbx_t \mid \mbx_0) \grad_{\mbx_t} \log p(\mbx_t \mid \mbx_0)$ and group $\frac{p(\mbx_t \mid \mbx_0) p(\mbx_0)}{p(\mbx_t)} = p(\mbx_0 \mid \mbx_t)$ to give
    \begin{align}
        &= \int p(\mbx_0 \mid \mbx_t) \grad_{\mbx_t} \log p(\mbx_t \mid \mbx_0) \ d\mbx_0 \\
        &= \int p(\mbx_0 \mid \mbx_t) \frac{\sqrt{\alpha_t} \mbx_0 - \mbx_t}{1 - \alpha_t} \ d\mbx_0 \\
        &=  \frac{\sqrt{\alpha_t} \E[\mbx_0 \mid \mbx_t] - \mbx_t}{1 - \alpha_t}.
    \end{align}
    Rearranging, we have our final result that gives the posterior mean as a function of the unconditional score function.
\end{proof}

\subsection{Signal Processing}

The Fourier transform is a fundamental tool in signal processing that decomposes a signal into its constituent frequencies. For a continuous signal $\mbx(t)$, the Fourier transform is given by
\begin{equation}
\mc{F}_{\text{cont}}(\mbx)[f] = \int_{-\infty}^{\infty} \mbx(t)e^{-2\pi i f t}dt,
\end{equation}
where $i$ denotes the complex root of unity. In discrete domains, such as digital images, we work with the Discrete Fourier Transform (DFT). For a signal $\mbx \in \R^n$, we denote its DFT as $\mathcal{F}(\mbx)$, where $\mathcal{F}(\mbx)[f_k]$ represents the frequency component at the k-th frequency
\begin{equation}
\mathcal{F}(\mbx)[f_k] = \sum_{n=0}^{N-1} \mbx[n]e^{-2\pi i k n/N}.
\end{equation}
The DFT can be expressed as a matrix operation $\mbf{F} \in \C^{n \times n}$ where
\begin{equation}
\mbf{F}_{jk} = \frac{1}{\sqrt{n}}e^{-2\pi i jk/n}.
\end{equation}
This matrix $\mbf{F}$ is unitary. A discrete convolution operation $\mbx \circledast \mbf{h}$ can be represented as a matrix multiplication $\mbf{C}_h \mbx$, where $\mbf{C}_h$ is a circulant matrix constructed from the filter $\mbf{h}$. A circulant matrix has the special property that each row is a cyclic shift of the previous row:
\begin{equation}
\mbf{C}_h = \begin{bmatrix}
h_0 & h_{n-1} & \cdots & h_1 \\
h_1 & h_0 & \cdots & h_2 \\
\vdots & \vdots & \ddots & \vdots \\
h_{n-1} & h_{n-2} & \cdots & h_0
\end{bmatrix}.
\end{equation}
A fundamental property of circulant matrices is that they can be diagonalized by the DFT matrix
\begin{equation}
\mbf{C}_h = \mbf{F}^* \text{diag}(\mathcal{F}(\mbf{h})) \mbf{F}.
\end{equation}
This relationship explains why convolution in the spatial domain equals pointwise multiplication in the frequency domain
\begin{equation}
\mathcal{F}(\mbx \circledast \mbf{h}) = \mathcal{F}(\mbx) \odot \mathcal{F}(\mbf{h}).
\end{equation}
Natural images typically exhibit a power law relationship in their frequency spectrum such that
\begin{equation}
|\mathcal{F}(\mbx)[f_k]|^2 \propto \frac{1}{|f_k|^\alpha},
\end{equation}
where $\alpha$ is typically around 2. This relationship, often called the $1/f^2$ law, arises from the fundamental structure of natural scenes:

\begin{itemize}
    \item Natural images tend to be locally smooth with occasional sharp transitions (edges)
    \item Objects in natural scenes exhibit self-similarity across scales
    \item Natural scenes contain hierarchical structures from fine to coarse details
\end{itemize}

\noindent This power law relationship provides a strong prior for image processing tasks, as it captures the statistical regularities present in natural images. The decay of frequency components according to this law explains why natural images are compressible and why high-frequency noise is particularly noticeable in image data.

\section{Proofs for \Cref{sec:limitations}} \label{sec:app:proofs}

We first prove Lemma \ref{lem:wiener-khinchin}, restated below.

\begin{lemma} \label{lem:wiener-khinchin_app}
    Let $f_k = \frac{k}{n}$ for $k = 0, \dots, n - 1$ denote the DFT sample frequencies for a signal of length $n$. There exists a covariance matrix $\bm{\Sigma}$ such that for $x \sim \mc{N}(\mbf{0}, \bm{\Sigma})$, the following two properties hold. First, the signal $\mbx$ follows a power law in the frequency domain with parameters $c, \beta > 0$ i.e. it has a power spectral density $S(f_k) = c |f_k|^{-\beta}$ for the non-zero DFT sample frequencies $f_k$. Second, the eigenvalues of $\bm{\Sigma}$ are precisely $c |f_k|^{-\beta}$.
\end{lemma}

\begin{proof}
    We will first provide a construction for $\bm{\Sigma}$. Let $\mbf{R}$ be a $n$-length signal that is the inverse Discrete Fourier Transform of $S(f_k) = c |f_k|^{-\beta}$ for the non-zero DFT sample frequencies $f_k$. We construct $\bm{\Sigma}$ as a circulant matrix whose first row (and column) is $\mbf{R}$.

    Next, we show the two properties of $\bm{\Sigma}$ needed to prove the lemma. First, let $\mbx \sim N(\mbf{0}, \bm{\Sigma})$. Then, $\mbx$ can be viewed as a finite stochastic process that is zero-mean, wide-sense, and stationary. It is zero-mean trivially because $\mbx$ is sampled from a zero-mean distribution. It is wide-sense stationary because $\bm{\Sigma}$ is circulant. Thus, the autocovariance function, which is exactly $\mbf{R}$, depends only on the gap between two elements in the signal. From the discrete-time Wiener-Khinchin theorem, we have that
    \begin{equation} \label{eq:wiener-khinchin_app}
        \E \left[  \frac{|\mc{F}(\mbx)[f_k] | ^2}{n} \right] = \mc{F}(\mbf{R})[f_k]
    \end{equation}
    By construction, we have that $\mc{F}(\mbf{R})[f_k] = S(f_k)$. This shows that in expectation, $\mbx$ follows a power spectral density of $S(f_k)$. Lastly, because $\bm{\Sigma}$ is circulant, we have that the eigenvalues of $\bm{\Sigma}$ are the Discrete Fourier Transform of the first row, which is $\mbf{R}$. As before, by construction, we have that $\mc{F}(\mbf{R})[f_k] = S(f_k)$. From \Cref{eq:wiener-khinchin_app}, we have that the eigenvalues of $\bm{\Sigma}$ are precisely $c |f_k|^{-\beta}$.
\end{proof}

Before we prove our main result, \Cref{thm:approx_gap}, we prove a useful lemma that calculates the posterior denoising distribution under a multivariate Gaussian assumption on the data. 

\begin{lemma} (Posterior Denoising Distribution) \label{lem:denoising_app}
    Suppose $\mbx_0\sim \mc{N}(\mbf{0}, \bm{\Sigma})$. Suppose $\mbx_t \mid \mbx_0 \sim \mc{N}(\sqrt{\alpha_t} \mbx_0, (1 - \alpha_t) \mbf{I})$. Then, 
    \begin{equation}
        p(\mbx_0 \mid \mbx_t) = \mc{N}(\tweedie, \bm{\Sigma}_{0 \mid t})
    \end{equation}
    where $\tweedie = \bm{\Gamma}_t \mbx_t$, $\bm{\Gamma}_t = \sqrt{\alpha_t} \bm{\Sigma} (\alpha_t \bm{\Sigma} + (1 - \alpha_t)\mbf{I})^{-1}$, and $\bm{\Sigma}_{0 \mid t} = \bm{\Sigma} - \sqrt{\alpha_t} \bm{\Gamma}_t \bm{\Sigma}$. 
\end{lemma}

\begin{proof}
    Suppose $\mbx_0 \sim \mathcal{N}(\mbf{0}, \bm{\Sigma})$ and $\mbx_t \mid \mbx_0 \sim \mathcal{N}(\sqrt{\alpha_t} \mbx_0, (1 - \alpha_t) \mbf{I})$. This implies that
    \begin{equation}
        \mbx_t = \sqrt{\alpha_t} \mbx_0 + (1 - \alpha_t) \bm{\epsilon}
    \end{equation}
    where $\bm{\epsilon} \sim \mc{N}(\mbf{0}, \mbf{I})$. Therefore, we can write the joint distribution of $\mbx_0$ and $\mbx_t$ as a multivariate Gaussian \cite{bishop2006pattern}:
    \begin{equation}
    \begin{bmatrix} \mbx_0 \\ \mbx_t \end{bmatrix} \sim \mc{N} \left(\mbf{0}, \begin{bmatrix} \bm{\Sigma} & \sqrt{\alpha_t} \bm{\Sigma} \\ \sqrt{\alpha_t} \bm{\Sigma} & \alpha_t \bm{\Sigma} + (1 - \alpha_t) \mbf{I} \end{bmatrix} \right).
    \end{equation}
    Using properties of conditional Gaussian distributions, we have that $p(\mbx_0 \mid \mbx_t)$ is also a Gaussian distribution with conditional mean and covariance
    \begin{equation}
        \E[\mbx_0 \mid \mbx_t] = \sqrt{\alpha_t} \bm{\Sigma} \left(\alpha_t \bm{\Sigma} + (1 - \alpha_t) \mbf{I} \right)^{-1} \mbx_t
    \end{equation}
    and 
    \begin{equation}
        \Cov [\mbx_0 \mid \mbx_t] = \bm{\Sigma} - \sqrt{\alpha_t} \bm{\Sigma} \left(\alpha_t \bm{\Sigma} + (1 - \alpha_t) \mbf{I} \right)^{-1} \sqrt{\alpha_t} \bm{\Sigma}.
    \end{equation}
    Letting $\bm{\Gamma}_t = \sqrt{\alpha_t} \bm{\Sigma} (\alpha_t \bm{\Sigma} + (1 - \alpha_t)\mbf{I})^{-1}$, $\tweedie = \E[\mbx_0 \mid \mbx_t]$ and $\bm{\Sigma}_{0 \mid t} = \Cov[\mbx_0 \mid \mbx_t]$, we have shown the lemma.
\end{proof}

Next, we prove our main theoretical result, \Cref{thm:approx_gap}, restated below.

\begin{theorem} \label{thm:approx_gap_app}
Suppose $\mbx_0$ is drawn from $\mc{N}(\mbf{0}, \bm{\Sigma})$ and we are given linear measurements $\mby = \mc{A}(\mbx_0)+ \mbf{z}$, where $\mc{A}(\mbx_0) = \mbA \mbx_0$ and $\mbf{z} \sim \mc{N}(0, \sigma_y^2 I)$. Suppose that the intermediate value $\mbx_t$ of the continuous-time reverse diffusion from \Cref{eq:reverse_cond_sde} can be written as $\mbx_t = \sqrt{\bar{\alpha}_t} \mbx_0 + \sqrt{1 - \bar{\alpha}_t} \epsilon$ where $\epsilon \sim \mc{N}(\mbf{0}, \mbf{I})$.  Then, we have that the true noisy likelihood score $\score{p(\mby \mid \mbx_t)}$ is
     \begin{equation} \label{eq:true_cond_score_app}
         \!\!\!
        (\mbA \bm{\Gamma}_t)^T \!(\mbA \bm{\Sigma}_{0 \mid t} \mbA^T \!+ \sigma_y^2\bm{I})^{-1} ( \mby - \mbA \tweedie ),\!\!
     \end{equation}
where $\bm{\Gamma}_t = \sqrt{\bar{\alpha}_t} \bm{\Sigma} (\bar{\alpha}_t \bm{\Sigma} + (1 - \bar{\alpha}_t) \bm{I})^{-1}$, $\bm{\mu}_{0 \mid t} = \E[\mbx_0 \mid \mbx_t] = \bm{\Gamma}_t \mbx_t$ and $\bm{\Sigma}_{0 \mid t} = \Cov [\mbx_0 \mid \mbx_t] = \bm{\Sigma} - \sqrt{\bar{\alpha}_t} \bm{\Gamma}_t \bm{\Sigma}$. Moreover, the FGPS approximation $\score{p(\mbf{C}_t \mby \mid \tweedie)}$ can be analytically calculated as
     \begin{equation} \label{eq:our_cond_score_app}
         \left(\mbf{C}_t \mbA \bm{\Gamma}_t \right)^T (\sigma_y^{2} \mbf{C}_t \mbf{C}_t^T)^{-1} (\mbf{C}_t \mby - \mbf{C}_t \mbA \tweedie) .
     \end{equation}
\end{theorem}

\begin{proof}
    We denote $\overleftarrow{\mbx_t}$ as the iterate from the continuous-time reverse diffusion process such that $\overleftarrow{\mbx_t} = \sqrt{\bar{\alpha_t}} \mbx_0 + \sqrt{1 - \bar{\alpha_t}} \epsilon$ where $\epsilon \sim \mc{N}(\mbf{0}, \mbf{I})$. From Lemma \ref{lem:denoising_app}, we have that $p(\mbx_0 \mid \overleftarrow{\mbx_t}) = \mc{N}(\tweedie, \bm{\Sigma}_{0 \mid t})$ where $\tweedie = \bm{\Gamma}_t \overleftarrow{\mbx_t}$ and $\bm{\Sigma}_{0 \mid t} = \bm{\Sigma} - \sqrt{\alpha_t} \bm{\Gamma}_t \bm{\Sigma}$. Now, similar to the proof of Lemma \ref{lem:denoising_app}, since $p(\mby \mid \mbx_0) = \mc{N}(\mbA \mbx_0, \sigma_y^2 \mbf{I})$, we can also calculate $p(\mby \mid \overleftarrow{\mbx_t})$ in closed form as another Gaussian distribution. Specifically, first we can write the joint distribution of $\mbx_0$ and $\mby$ conditioned on $\overleftarrow{\mbx_t}$ as a multivariate Gaussian \cite{bishop2006pattern}:
    \begin{equation}
        \begin{bmatrix} \mbx_0 \\ \mby \end{bmatrix} \Big| \overleftarrow{\mbx_t} \sim \mc{N} \left( \begin{bmatrix} \tweedie \\ \mbA \tweedie \end{bmatrix}, \begin{bmatrix} \bm{\Sigma}_{0 \mid t} & \bm{\Sigma}_{0 \mid t} \mbA^T \\ \mbA \bm{\Sigma}_{0 \mid t} & \mbA \bm{\Sigma}_{0 \mid t} \mbA^T + \sigma_y^2 \mbf{I} \end{bmatrix} \right).
    \end{equation}
    Further, using properties of conditional Gaussian distributions, we have that $p(\mby \mid \overleftarrow{\mbx_t}) = \mc{N}(\mbA \tweedie, \mbA \bm{\Sigma}_{0 \mid t} \mbA^T + \sigma_y^2 \mbf{I})$. Let $\bm{\Delta}_t = \mby - \mbA \tweedie$. Then, computing the gradient of $p(\mby \mid \overleftarrow{\mbx_t})$ with respect to $\overleftarrow{\mbx_t}$, we have
    \begin{align}
        \!\!\ \grad_{\overleftarrow{\mbx_t}} \log p(\mby \mid \overleftarrow{\mbx_t}) &= \grad_{\overleftarrow{\mbx_t}} \log \mc{N}(\mbA \tweedie, \mbA \bm{\Sigma}_{0 \mid t} \mbA^T + \sigma_y^2 \mbf{I}) \\
        &= \grad_{\overleftarrow{x_t}} -0.5 \bm{\Delta}_t^T (\mbA \bm{\Sigma}_{0 \mid t} \mbA^T + \sigma_y^2 \mbf{I})^{-1} \bm{\Delta}_t  \\
        &= \left( \mbA \frac{\partial \tweedie}{\partial \overleftarrow{x_t}} \right)^T (\mbA \bm{\Sigma}_{0 \mid t} \mbA^T + \sigma_y^{2}\mbf{I})^{-1} \bm{\Delta}_t 
    \end{align}
    As $\tweedie = \bm{\Gamma}_t \overleftarrow{\mbx_t}$, we have that $\frac{\partial \tweedie}{\partial \overleftarrow{\mbx_t}} = \bm{\Gamma}_t$, which gives us \Cref{eq:true_cond_score_app}. The FGPS approximation to the true conditional score is $\grad_{\overleftarrow{\mbx_t}} \log p(\mby \mid \overleftarrow{\mbx_t}) \approx \grad_{\overleftarrow{\mbx_t}} \log p(\mbf{C}_t \mby \mid \tweedie) = \grad_{\overleftarrow{\mbx_t}} \mc{N}(\mbA \tweedie, \sigma_y^2 \mbf{I})$. Similar to above, we can also calculate its gradient $\grad_{\overleftarrow{\mbx_t}} \log p(\mbf{C}_t \mby \mid \tweedie)$ with respect to $\overleftarrow{\mbx_t}$ as
    \begin{equation} \label{eq:fgps_approx_app}
        \!\!\ \left(\mbf{C}_t \mbA \bm{\Gamma}_t \right)^T (\sigma_y^{2} \mbf{C}_t \mbf{C}_t^T)^{-1} (\mbf{C}_t \mby - \mbf{C}_t \mbA \tweedie).
    \end{equation}
\end{proof}

Corollary \ref{cor:dps_approx} can easily be proven by taking $\mbf{C}_t$ as the identity matrix in the above proof.

\begin{figure*}[htbp]
    \centering
    \includegraphics[width=0.6\textwidth]{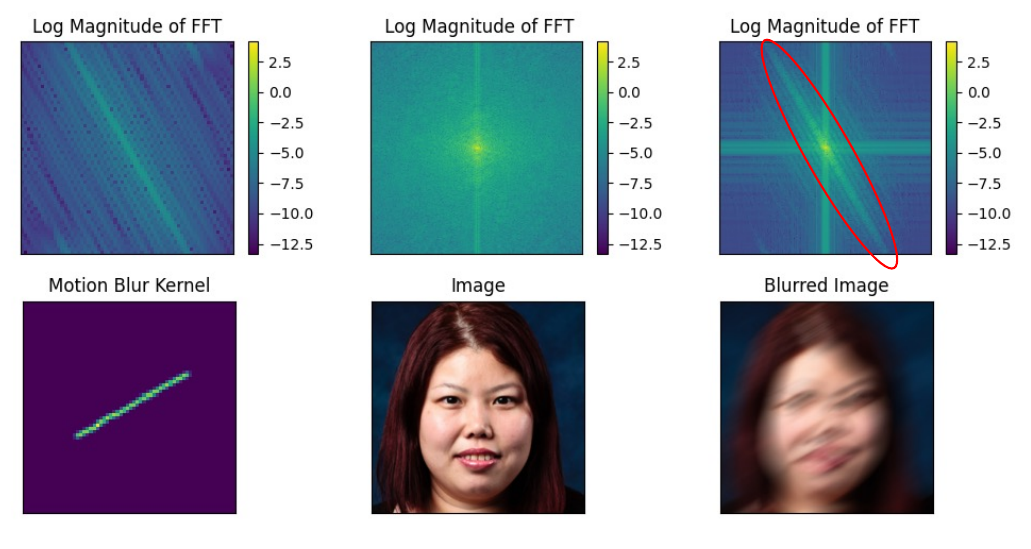}
    \caption{Directional motion blur can retain high frequency components in certain spatial directions orthogonal to the direction of the blur. The red circled direction shows an example of such a direction.}
    \label{fig:blur_fft}
\end{figure*}

\section{Theoretical Investigation of Approximation Gap of FGPS and DPS} \label{sec:app:theory}

Our main theorem shows in the multivariate Gaussian setting, the true conditional score differs from approximations in the term \( (\mbf{I} + \mbf{A} \bm{\Sigma}_{0 \mid t} \mbf{A}^\top)^{-1} \), which requires approximations in general. DPS approximates it at as the identity matrix. In this section, we show that in certain cases, FGPS is a significantly better approximation than the identity matrix. To see this, consider when $\mbA$ is a high-pass filter and the data covariance is $\bm{\Sigma}_f$ such that it follows a power law in the frequency domain. Then, the matrix $\mbA \bm{\Sigma}_{0 \mid t} \mbA^\top$ has significant energy in high-frequency directions. Specifically, in the Fourier basis, denoting the eigenvalues of $\mbA$, $\bm{\Sigma}_f$, $\bm{\Sigma}_{0 \mid t}$ as $a_i$, $\lambda_i$ and $\lambda_i^{t}$ respectively, the eigenvalues of  $\mbA \bm{\Sigma}_{0 \mid t} \mbA^\top$ take the form $a_i^2 \lambda_i^{t}$,  where $a_i^2$ grows with frequency (due to the high-pass nature of $\mbA$), and $\lambda_i^{t} = \frac{(1 - \alpha_t)\lambda_i}{\alpha_t \lambda_i + (1 - \alpha_t)}$. Although  $\lambda_i^{t}$ is small in high-frequency directions, the product $a_i^2 \lambda_i^{\text{post}}$ can still be $\mathcal{O}(1)$ or larger due to the amplification by $\mbA$. Consequently, the eigenvalues of $(\mbf{I} + \mbA \bm{\Sigma}_{0 \mid t} \mbA^\top)^{-1}$ in those directions, given by $\frac{1}{1 + a_i^2 \lambda_i^{\text{post}}}$, are significantly less than $1$, indicating strong suppression of high-frequency components especially when $\alpha_t$ is very small as in the beginning of the reverse process. In contrast, low-frequency directions (where $a_i^2 \approx 0$) are preserved. This shows that $(\mbf{I} + \mbA \bm{\Sigma}_{0 \mid t} \mbA^\top)^{-1}$ acts as a low-pass filter. The FGPS approximation, $\mbf{C}_t^\top (\mbf{C}_t \mbf{C}_t^\top)^{-1} \mbf{C}_t$, which is a projection matrix onto the low-frequency components, thus approximates this behavior more faithfully than the identity matrix, which uniformly preserves all directions.

\section{Further Motivations for our Method} \label{sec:app:further_motivation}

In \Cref{sec:limitations}, we argued that when the forward operator is convolution with a high-pass filter, existing methods have a large approximation between the conditional score and its approximation, which can lead to compromised sample quality in practice. High-pass filtering may seem like a contrived example, because after all, many inverse problem tasks considered in the literature have forward operators that are low-pass filters, such as Gaussian deblurring and superresolution tasks. Even though some natural imaging systems such as phase contrast microscopy, we argue that the high-pass filter effect can also show up in more complex image restoration tasks. For example, for motion deblurring, the forward operator can act as a high-pass filter in certain spatial directions of the image. In \Cref{fig:blur_fft}, we examine this effect by looking at the log magnitude of the frequency domain of an image convolved with a simple directional blur kernel. We can clearly see that in the red circled direction of the Fourier domain, the filter retains high frequency components of the original image. This is also mathematically evident using the Fourier convolution theorem. The Fourier convolution theorem states that

\begin{equation} \label{eq:fourier_conv}
    \mc{F}(k \circledast x) = \mc{F}(k) \cdot \mc{F}(x),
\end{equation}

where $\circledast$ denotes convolution and $\cdot$ denotes an element-wise product. Further, for a directional blur kernel, the Fourier transform in spatial directions has high frequency values in directions orthogonal to the direction of the blur. Thus, it is clear that in those directions, the motion blur will retain high frequency components. While the high-pass filter considered in \Cref{fig:approx_gap} was an extreme case of this, our analysis highlights a crucial deficiency of existing methods since high frequency components of the measurement can amplify approximation errors. 

Lastly, we emphasize that our experimental results demonstrate a fascinating phenomenon where the performance gap between FGPS and baseline methods is even larger on more complex datasets, hinting at the fact that besides the frequency characteristics of the forward operator, the frequency characteristics and quality of the Tweedie estimate also greatly affects reconstruction. It would be an interesting theoretical direction to understand why the frequency schedule has larger benefits in this case. 

\section{Experimental Details} \label{sec:app:experimental_details}

Below, we list the detailed setup for all experiments reported in the main paper. All the images for both FFHQ and ImageNet dataset are resized to $256 \times 256$, and we report results for 1000 images from the validation datasets for the FFHQ and Imagenet datasets. We will release our code upon publication.

\subsection{Inverse Problem Task Descriptions}

\subsubsection{Linear Inverse Problems}

For the first three image restoration problems we consider, the forward operators are linear operators defined as convolution with a given kernel. These kernels are all of size $61 \times 61$. 

\noindent \textbf{Gaussian Deblurring.} The forward operator is convolution with a Gaussian blur kernel with standard deviation $3$. 

\noindent \textbf{Motion Deblurring.} The forward operator is convolution with a motion blur kernel generated from code\footnote{\url{https://github.com/LeviBorodenko/motionblur}} with intensity $0.5$. 

\noindent \textbf{Deconvolution with High Pass Filter.} the forward operator is a high-pass filter implemented as a Dirac kernel minus a Gaussian blur kernel of standard deviation $5.0$. 

\subsubsection{Nonlinear Inverse Problems}

\noindent \textbf{Image Dehazing.} The fourth image restoration is a nonlinear inverse problem, image dehazing, where the forward operator is a hazing operator with strength $1$. The hazing operator models how light is scattered and attenuated in the atmosphere, based on the atmospheric scattering model
\begin{equation}
 \mc{A}(\mbx) = \mbx \cdot t(\mbx) + L(1-t(\mbx)),
\end{equation}
where $\mbx$ is the original clear image (scene radiance), $t(\mbx)$ is the transmission map representing the portion of light that reaches the camera, and $L$ is the atmospheric light value. The transmission map $t(\mbx)$ can be modeled using the Beer-Lambert law, which states that
\begin{equation}
t(\mbx) = e^{-\beta d(\mbx)}.
\end{equation}
Above, $\beta$ is the atmospheric scattering coefficient, and $d(\mbx)$ is the scene depth map. In our experiments, we set $L = 1, \beta = 1$, and set $d(\mbx)$ be Euclidean distance of each pixel to the center of the image. 

\subsection{Baselines}

\noindent \textbf{DPS} \cite{chung2022diffusion} We use the code from \url{https://github.com/DPS2022/diffusion-posterior-sampling} using the default hyperparameter settings for Gaussian deblurring and motion deblurring for both FFHQ and ImageNet. For high-pass filter operator, we used the same hyperparameters used for motion deblurring.

\noindent \textbf{MCG} \cite{chung2022improving} We use the code from \url{https://github.com/HJ-harry/MCG_diffusion} using the default hyperparameter settings for Gaussian deblurring and motion deblurring for both FFHQ and ImageNet. For high-pass filter operator, we used the same hyperparameters used for motion deblurring.

\noindent \textbf{DSG} \cite{yang2024guidance} We use the code from \url{https://github.com/LingxiaoYang2023/DSG2024} using the default hyperparameter settings for Gaussian deblurring for both FFHQ and ImageNet. The only hyperparameter change was that we set $interval = 10$ as we observed this worked better in practice. For high-pass filter deconvolution and motion deblurring operator, we used the same hyperparameters used for Gaussian deblurring.

\noindent \textbf{Score-SDE/ILVR} \cite{song2020score,choi2021ilvr} Generally, we group Score-SDE and ILVR as methods that use a sequence of noisy measurements to approximate the conditional score, as mentioned in \cite{daras2024survey} and \cite{chung2022diffusion}. This is a generalization of the methods in \cite{song2020score} and \cite{choi2021ilvr} as their methods only were presented for the inpainting and superresolution tasks. To consider general tasks, we sample a sequence of measurements $\mby_t \sim \mc{N}(\sqrt{\bar{\alpha_t}} \mby, (1 - \bar{\alpha_t} \mbf{I}))$. Then, we approximate the conditional score as $-\eta \grad_{\mbx_t} \norm{\mby_t - \mc{A}(\mbx_t)}{2}^2$. We use the step size $\eta$ from the DPS method \cite{chung2022diffusion}. 

\noindent \textbf{AOD-Net.} We use the code from \url{https://github.com/MayankSingal/PyTorch-Image-Dehazing} and train the network using the default parameters on 10000 images from the FFHQ training dataset. 

\noindent \textbf{DoubleDIP.} We use the code from \url{https://github.com/yossigandelsman/DoubleDIP} using all default parameters. This method is unsupervised and does not require any retraining.

\subsection{Details for \Cref{fig:approx_gap}}

In \Cref{fig:approx_gap}, we demonstrated the approximation gap of DPS and our method on synthetic data. Below, we describe the precise experimental details of our experiments.

\noindent \textbf{Forward Operators.} We consider a high-pass filter operators in our experiments. a low-pass filter and a high-pass filter. First, we create a Gaussian blur kernel of width $\sigma$ i.e. this kernel is $e^{-\mbx^2 / (2 \sigma^2)}$ and normalized to have sum $1$. For the high-pass filter, we simply subtract the Gaussian blur kernel of width $\sigma$ from a Dirac kernel, which has 1 at the center and zeroes elsewhere. The filter is a circulant matrix constructed from this kernel. Importantly, the first row of the circulant matrix corresponding to this kernel is normalized to have sum zero, such that the kernel is indeed a high-pass filter.

\noindent \textbf{Data Generation.} Next, we describe the data generation process. Recall from Lemma \ref{lem:wiener-khinchin} the construction of a covariance matrix $\bm{\Sigma}_f$ such that data drawn from $\mc{N}(\mbf{0}, \bm{\Sigma}_f)$ follows a power law in the frequency domain with parameters $c$ and $\beta$. We let $c = 1$ and $\beta = 2.5$, which is the typical range for natural image data \cite{van1996modelling}. We draw 10000 signals $\mbx^{(i)}$ of length 2000 from $\mc{N}(\mbf{0}, \bm{\Sigma}_f)$. For each signal, we generate a corresponding measurement $\mby^{(i)} = \mbA \mbx^{(i)} + \mbf{z}$, where $\mbf{z} \sim \mc{N}(\mbf{0}, \mbf{I})$, and $\mbA$ is the circulant matrix corresponding to the low-pass filter or the high-pass filter.

\noindent \textbf{Approximation Gap.} For each $\mbx^{(i)}, \mby^{(i)}$ pair, we calculate $\mbx_t^{(i)} \sim \mc{N}(\sqrt{\bar{\alpha_t}} \mbx^{(i)}, (1 - \bar{\alpha_t}) \mbf{I})$, where $\bar{\alpha_t}$ is the variance schedule from the DDPM paper \cite{ho2020denoising}. We report the average approximation gap over the 10000 signals. Below, we give the exact approximation that our method, FGPS, makes for the noisy likelihood score. From \Cref{thm:approx_gap}, we know that under a multivariate normal assumption on the data, the posterior mean $\tweedie$ is $\bm{\Gamma}_t \mbx_t$. Under the same assumptions as \Cref{thm:approx_gap}, for linear inverse problems where $\mc{A}(\mbx_0) = \mbA \mbx_0$ we analytically compute 
our approximation to the noisy likelihood score given in \Cref{eq:our_score_approx} as
\begin{equation} \label{eq:our_score_approx_full}
    (\mbf{C}_t \mbA)^T (\sigma_y^{2} \mbf{C}_t \mbf{C}_t^T)^{-1} (\mbf{C}_t \mby - \mbf{C}_t \mbA \tweedie).
\end{equation}

\noindent \textbf{Frequency Curriculum.} For our method, we have the choices of the frequency curriculum described through $\tau_T$ and $\tau_1$, the frequency cutoffs for the time-dependent low pass filter. For simplicity, we use a frequency curriculum from the variance schedule of the diffusion model. Specifically, let $\sigma_t = \frac{1}{\bar{\alpha_t}} - 1$, which corresponds to the corresponding variance-exploding parameterization of the SDE (see \cite{kawar2022denoising}, Appendix B). Then we set $\tau_t$ be the max frequency value $f_k$ such that $S(f_k) \geq \max\{\sigma_y^2, \sigma_t^2\}$, where $S(f_k)$ denotes the power spectral density from \Cref{eq:psd}. This utilizes the intuition that for $\mbx_t^{(i)}$, the frequency components of the original signal $\mbx^{(i)}$ still present in the image are the frequencies below $\tau_t$. 

\subsection{Implementation Details for Our Method}

In practice, we let $\tau_1$ and $\tau_T$ be hyperparameters and consider various schedules to interpolate between them. However, our theoretical results can guide us on setting them in a data-dependent way. Recall that natural images tend to follow a radially averaged power law in the frequency domain such that each non-zero frequency $f_k$ has power $S(f_k) \approx c |f_k|^{-\beta}$ for some constants $c, \beta > 0$. We employ three heuristics for setting our frequency schedule.

\begin{enumerate}
   \item (Setting $\tau_1$ to improve robustness to noise) When the measurement noise has variance $\sigma_y^2$, higher frequency components of $\mby$ are affected, specifically those frequencies $f$ such that $S(f) \leq \sigma_y^2$. Denoting $f_\text{{noise}}$ as the minimum frequency value that satisfies this constraint, we can set $\tau_1$ less than $f_\text{{noise}}$ and make our method robust to additive Gaussian noise of a known variance. 
   \item (Frequency Schedule as Function of Power Law Decay) For data that follows a faster decaying power law (when $\beta$ is larger), we can set a schedule for $\tau_t$ that also increases faster as a function of $t$. 
   \item (Setting $\tau_T$) For natural images, it is sufficient to set $\tau_T$ to be a small percentage of the overall frequency range, as this contains most of the information present in the measurement $\mby$. This is the same intuition behind JPEG compression. In our experiments on natural images, we take $\tau_T$ to be roughly 30\% of the overall frequency range. 
\end{enumerate}

To implement a frequency schedule, we use a binary mask applied in the frequency domain. Specifically, we utilize the Fast Fourier Transform (FFT) to transform the image data into the frequency domain. A low-pass filter mask is then created based on a specified cutoff frequency $\tau_t$. We utilize the Euclidean distance from the center of the frequency domain to create the low-pass filter mask, which selectively retains frequencies below the cutoff value. As the reverse process progresses, the cutoff is adjusted to allow more high-frequency components, thereby refining the image details. This method efficiently integrates frequency control into the reverse diffusion process, contributing to improved image restoration. In our work, we consider two frequency schedules, an exponential schedule and a linear schedule, which interpolate between $\tau_T$ and $\tau_1$ in different ways. Precisely, the exponential schedule follows 
\begin{equation}
    \tau_t = \tau_1 - (\tau_1 - \tau_T) \exp\left(-\frac{5t}{T}\right),
\end{equation}
and the linear schedule follows
\begin{equation}
    \tau_t = \tau_T + \frac{t}{T}(\tau_1 - \tau_T).
\end{equation}
Further, we set the step size $\mbf{S}_t$ to be $\mbf{S}_t = \frac{\kappa_t}{\norm{\phi_t(\mby) - \phi(\mc{A}(\tweedie))}{2}} \mbf{I}$ for a scalar time-dependent hyperparameter $\kappa_t$. We set this schedule to smoothly transition from $\kappa_T$ at the beginning of the reverse process to $\kappa_1$ at the end. Precisely, this schedule follows
\begin{equation} \label{eq:kappa_schedule}
    \kappa_t = \frac{1}{2} (\kappa_T + \kappa_1) + \frac{1}{2}(\kappa_T - \kappa_1) \cos \left( \frac{\pi t}{T} \right),
\end{equation}
where $t$ denotes the current time step, and $T$ is the total number of time steps (e.g. 1000 in our case).

Next, we detail all the hyperparameter settings used for our experiments on the FFHQ and ImageNet datasets.
 
    \begin{itemize}
        \item FFHQ
        \begin{itemize}
            \item Gaussian Deblurring:
            \begin{itemize}
                \item Step Size: $\kappa_T=3.0, \kappa_1=0.6$
                \item Frequency Schedule: Exponential
            \end{itemize}
            \item Motion Deblurring:
            \begin{itemize}
                \item Step Size: $\kappa_T=5.0, \kappa_1=1.0$
                \item Frequency Schedule: Exponential
            \end{itemize}
            \item High pass:
            \begin{itemize}
                \item Step Size: $\kappa_T=5.1, \kappa_1=1.1$
                \item Frequency Schedule: Linear
            \end{itemize}
            \item Haze:
            \begin{itemize}
                \item Step Size: $\kappa_T=5.0, \kappa_1=1.0$
                \item Frequency Schedule: Exponential
            \end{itemize}
        \end{itemize}
        
        \item ImageNet
        \begin{itemize}
            \item Gaussian Deblurring:
            \begin{itemize}
                \item Step Size: $\kappa_T=2.0, \kappa_1=0.01$
                \item Frequency Schedule: Linear
            \end{itemize}
            \item Motion Deblurring:
            \begin{itemize}
                \item Step Size: $\kappa_T=3.0, \kappa_1=0.1$
                \item Frequency Schedule: Linear
            \end{itemize}
            \item High pass:
            \begin{itemize}
                \item Step Size: $\kappa_T=3.5, \kappa_1=0.6$
                \item Frequency Schedule: Linear
            \end{itemize}
        \end{itemize}
    \end{itemize}

\noindent In Appendix \ref{sec:app:further_exp}, we provide ablation studies to understand the effects of these hyperparameters on the generation quality. 

\subsection{Computational Overhead of Our Method}
In this section, we evaluate the computational efficiency of our method by measuring the average runtime per image on the FFHQ dataset. We run each method from \Cref{tab:linear_inv} on 100 images, and report the average time taken per image, given in Table~\ref{tab:runtime}. We see that our method introduces minimal computational overhead compared to the DPS method, demonstrating that our modifications can be implemented in an efficient manner and easily integrated into existing frameworks.

\begin{table}[H]
\centering
\begin{tabular}{|l|c|}
\hline
\textbf{Method} & \textbf{Time (seconds)} \\
\hline
DSG & 45.13744 \\
Score-SDE/ILVR & 41.411510 \\
MCG & 93.20647 \\
DPS & 90.641704 \\
FPGS (Ours) & 92.383272 \\
\hline
\end{tabular}
\caption{Average Wall Clock Runtime Per Image (seconds) on FFHQ Dataset}
\label{tab:runtime}
\end{table}

\section{Comparison to ILVR/Score-SDE} \label{sec:app:ilvr}

Our method is closely related to the Score-SDE and ILVR works that consider a sequence of noisy measurements $y_t$ such that $y_t = \mc{N}(\sqrt{\alpha_t} y_0, (1 - \alpha) \mbf{I})$ \cite{choi2021ilvr, song2020score}. These methods consider approximations to the noisy likelihood score that are typically of the form $\mbf{L}_\mc{A} (\mby_t - \mc{A}(\mbx_t))$, where $\mbf{L}_\mc{A}$ is a fixed matrix that depends on the measurement operator $\mc{A}$. There are three crucial differences between these methods and our method. First, our approximation to the noisy likelihood score considers a time-varying model likelihood at each diffusion timestep $p(\mby \mid \tweedie)$ as given in \Cref{eq:new_likelihood} as opposed to simply varying $\mby_t$. Further, we leverage the powerful denoising capabilities of the pretrained diffusion model by using the Tweedie estimate $\tweedie$ instead of the noisy $\mbx_t$. Lastly, we propose a frequency curriculum that can differ from the variance schedule of the diffusion model and be adapted to the frequency characteristics of the data. These differences lead to a drastically improved performance of our method compared to Score-SDE/ILVR as seen in \Cref{tab:linear_inv}. 

\section{Further Experiments and Ablation Studies} \label{sec:app:further_exp}

\subsection{Visualizing the Transformed Measurements}

In \Cref{fig:cond_meas_app}, we show the transformed measurements at three timesteps in the reverse diffusion process when applying our frequency curriculum on the FFHQ dataset on all the forward operators we considered. Our theoretical results indicate that the initial reverse process steps are very important in order to obtain coarse alignment with the given measurement, so it is reasonable to align the measurements with this coarse-to-fine strategy.

\begin{figure*}[ht]
\centering
\includegraphics[width=0.7\textwidth]{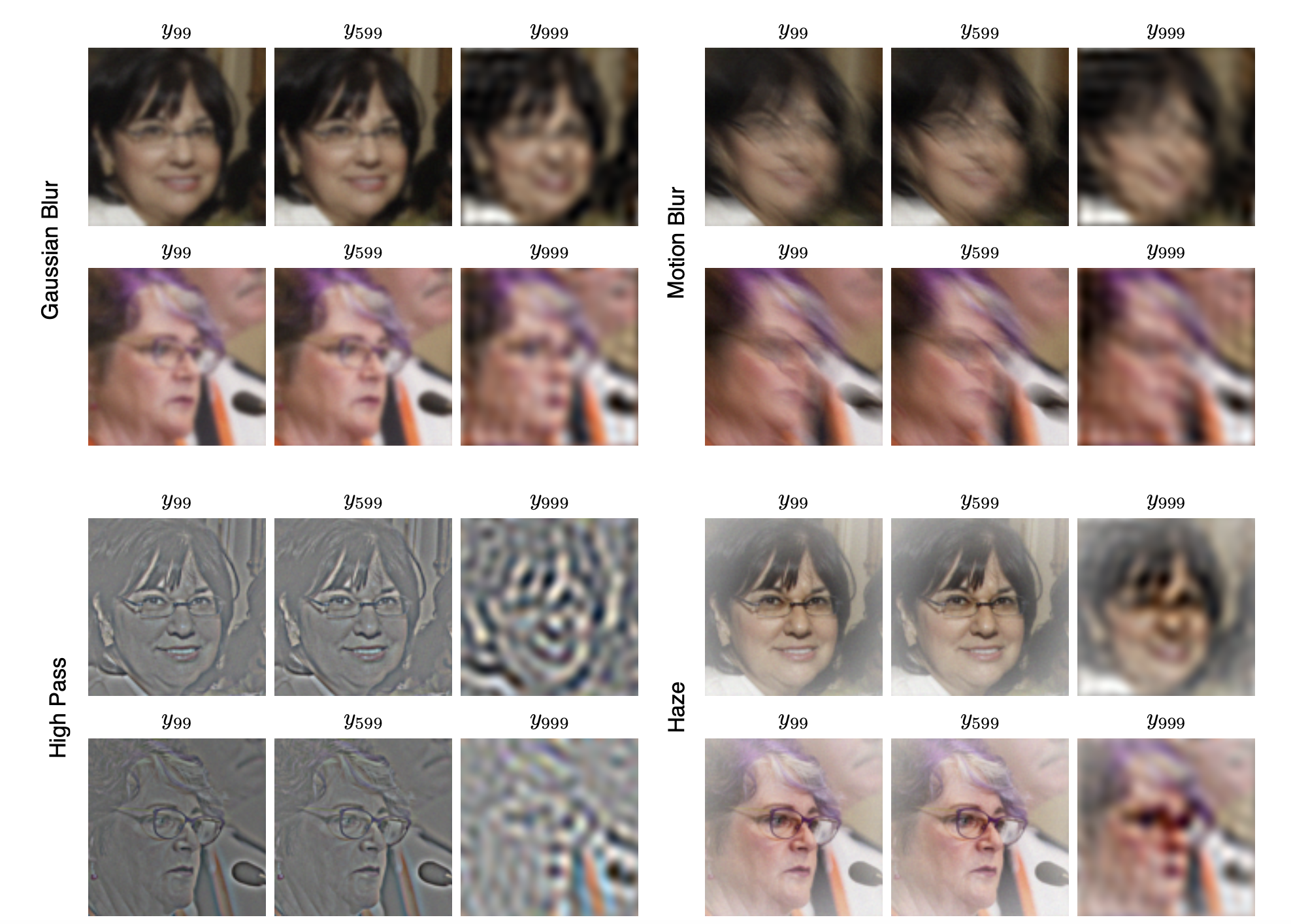} 
\caption{We demonstrate the visual effects of our frequency curriculum on the transformed measurements. In the early stages of the reverse process, only very coarse features of the measurement are retained.}
\label{fig:cond_meas_app}
\end{figure*}

\begin{figure*}[htbp]
\centering
\includegraphics[width=0.7\textwidth]{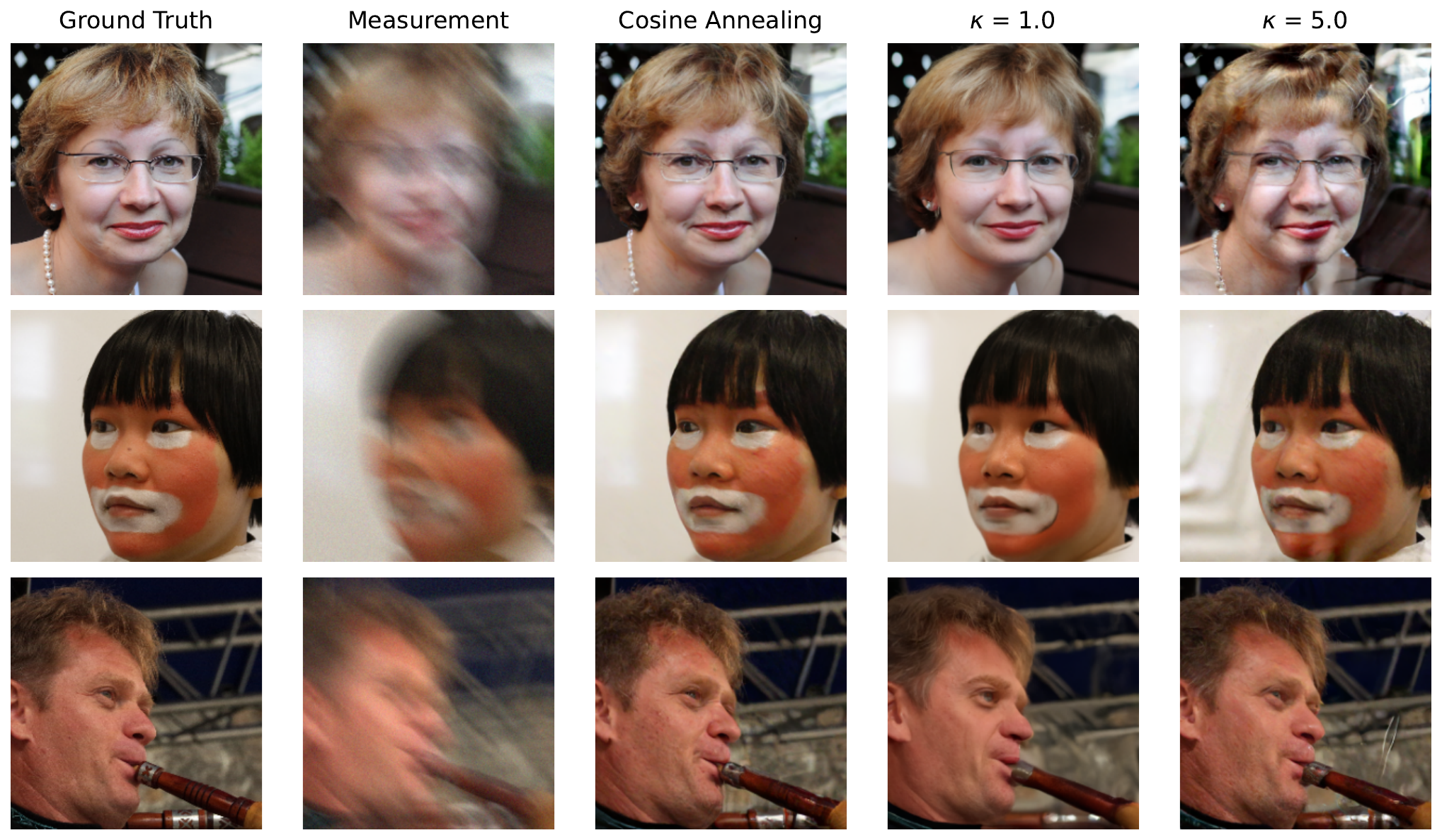} % Replace with the actual path to your image
\caption{Qualitative comparison of different scale schedules on the FFHQ dataset. The first row highlights how the cosine schedule accurately captures details, such as the necklace, while the fixed small step size ($\kappa=1$) fails and the fixed large step size ($\kappa=5$) introduces artifacts.}
\label{fig:scale_schedule_image}
\end{figure*}

\subsection{Effect of Frequency Curriculum} \label{sec:app:freq_curr}

\noindent \textbf{Impact of Time-Dependent Curriculum.} As observed in \Cref{fig:ablation_kt}, we observe that the time-dependent frequency curriculum helps the stability of the method. Namely, when the operator is a high pass filter, the time-dependent frequency curriculum results in high-quality reconstructions. On the other hand, for a fixed curriculum, we observe that over different generations, most random samples of the diffusion model fall off the natural image manifold in the early steps of the reverse process, which results in unnatural looking images as in \Cref{fig:ablation_kt}. 

Next, to assess the impact of different frequency curricula on image quality, we provide a visual comparison using images generated with linear and exponential time-dependent curricula on the FFHQ and ImageNet datasets.

\noindent \textbf{FFHQ Motion Deblurring.} The images produced with the exponential frequency schedule exhibit superior quality, as shown in Figure~\ref{fig:ffhq_freq}. The exponential increase in high-frequency components helps to better capture fine facial features and intricate details, resulting in more visually appealing images. 

\noindent \textbf{ImageNet Motion Deblurring.} The ImageNet dataset benefits more from the linear frequency schedule, as illustrated in Figure~\ref{fig:imagenet_freq}. We hypothesize this is because the diversity and detail of ImageNet scenes require more careful consideration at many frequency ranges to obtain a high-quality reconstruction.

\noindent \textbf{Motivation for Frequency Schedule Selection.} Our observations indicate that the choice of frequency schedule should be informed by the characteristics of the underlying data. The exponential schedule is well-suited for datasets with structured and hierarchical features, such as faces, where it is crucial to refine high-frequency details over an extended number of reverse process time steps. In contrast, datasets like ImageNet, which contain diverse and complex textures, benefit from a more uniform and gradual frequency progression, ensuring consistent detail incorporation throughout the sampling process. We note though that while our method offers the flexibility to incorporate different frequency curricula, it is not a necessity to obtain high-quality reconstructions, and even a simple exponential curriculum can already obtain results that are significantly better than baseline methods. We simply highlight that this can be further improved by carefully considering the frequency distribution of the data.

\begin{figure*}[hbp]
\centering
\includegraphics[width=0.8\textwidth]{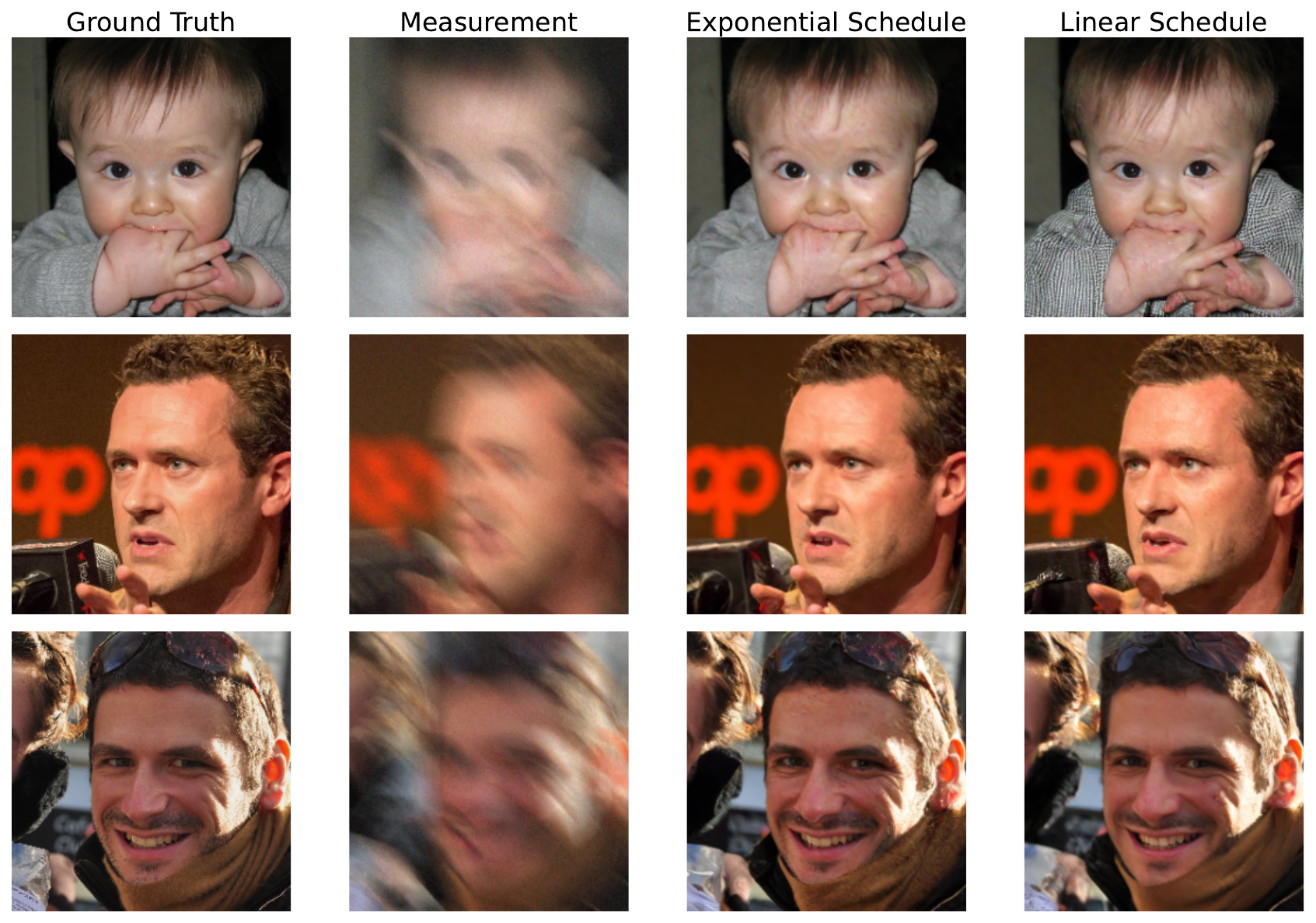}
\caption{ Comparison of images on the FFHQ dataset using exponential and linear frequency schedules. The exponential schedule produces higher-quality images with refined details and realistic textures compared to the linear schedule.}
\label{fig:ffhq_freq}
\end{figure*}

\begin{figure*}[htbp]
\centering
\includegraphics[width=0.8\textwidth]{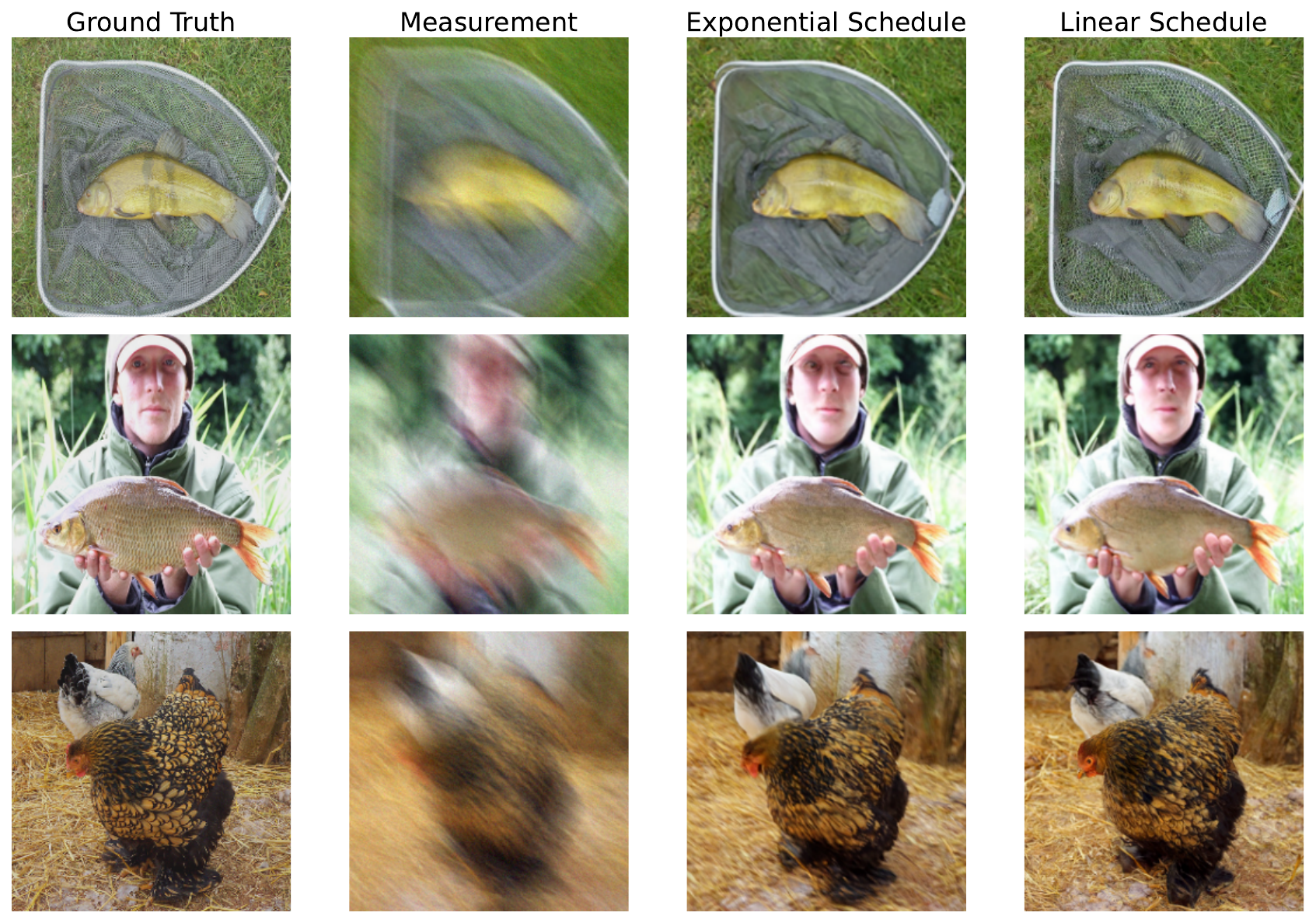} 
\caption{Comparison of images on the ImageNet dataset using exponential and linear frequency schedules. The linear schedule preserves object shapes and structural details better, resulting in clearer images.}
\label{fig:imagenet_freq}
\end{figure*}

\subsection{Effect of Step Size Schedule}
Recall that in \Cref{eq:our_score_approx}, we had that our approximation to the score should be scaled by $\mbf{S}_t$, which can be thought of as a step size to the gradient terrm. In practice, instead of using this exact quantity, which would require using a potentially large circulant matrix and its inverse, we set $\mbf{S}_t = \frac{\kappa}{\norm{\phi_t(\mby) - \phi(\mc{A}(\tweedie))}{2}} \mbf{I}$ for a scalar hyperparameter $\kappa$, which works well in practice. This is similar to the step size chosen in DPS. The intuition behind this step size choice is to control the approximation error. For example, when the approximation error is high, the step size should be smaller as the conditional score is noisy. These schedules play a crucial role in the effectiveness and stability of the optimization process. In this section, we study the effect of different step size schedules to motivate our choice. Specifically, we study three step size schedules (where we vary $\kappa$) on the motion deblurring task on the FFHQ dataset. These three schedules are: 

\begin{enumerate}
    \item Fixed small step size (used by DPS)
    \item Fixed large step size
    \item Cosine Annealed step size according to \Cref{eq:kappa_schedule}. This schedule starts with a relatively high step size, which facilitates rapid progress in the initial stages, and gradually decreases to smaller step sizes as the optimization nears convergence. 
\end{enumerate}

In \Cref{tab:scale_schedule} and \Cref{fig:scale_schedule_image}, we demonstrate the quantitative and qualitative differences between the three schedules. We clearly see that the cosine annealed step size strongly outperforms the other two schedules, which is why we adopt it for our experiments. We hypothesize that the cosine schedule allows the model to quickly capture the coarse structure of the image early on, while the progressively smaller steps enable fine-tuning of details, leading to a more refined final reconstruction. This is evident from \Cref{fig:scale_schedule_image}. For example, in the first row, the cosine schedule successfully captures the necklace detail around the subject's neck, while the fixed small step size schedule completely misses this feature, resulting in a blurred and oversmoothed output. The fixed large step size schedule manages to retain some necklace details but introduces noticeable artifacts, which degrade the overall image quality. The large step size likely destabilizes the optimization and pushes the image off the natural image manifold, which compromises image quality.

\begin{table}[H]
\centering
\small
\begin{tabular}{|l|c|c|c|c|}
\hline
\textbf{Step Size Schedule} & \textbf{FID$\downarrow$} & \textbf{LPIPS$\downarrow$} & \textbf{PSNR$\uparrow$} & \textbf{SSIM$\uparrow$} \\
\hline
Cosine Annealing & \textbf{49.66} & \textbf{0.1254} &\textbf{ 25.53}  & \textbf{0.724}  \\
Fixed Lower Bound & 70.73 & 0.1790& 24.22  & 0.677 \\
Fixed Upper Bound & 59.94 & 0.1559 & 24.22 & 0.686\\
\hline
\end{tabular}
\caption{Performance comparison of different scale schedules on the FFHQ dataset. Metrics include FID, LPIPS, PSNR, and SSIM.}
\label{tab:scale_schedule}
\end{table}

\subsection{Further Qualitative Results}

\noindent \textbf{Comparison to Baselines.}  In \Cref{fig:main_results_zoom}, we see the zoomed in version from \Cref{fig:main_results}, highlighting the intricate details of the image captured by our reconstructions. In \Cref{fig:motion_baseline_grid}, we see a comprehensive comparison to all the baselines reported in \Cref{tab:linear_inv} on the motion blurring task on the FFHQ dataset. We observe that Score-SDE/ILVR often captures the correct structure of the image but fills in different facial features. MCG usually overfits to the Gaussian measurement noise as reported in \cite{chung2022diffusion}. DPS results in good quality reconstructions, but the images are usually smoothed out and lose small facial features. DSG performs very well on the motion deblurring task, as shown also in \Cref{tab:linear_inv}. However, DSG often gives grainy reconstructions and artifacts that are clearly visible. In contrast, our method is significantly more stable and is able to give images that have higher perceptual quality than all the baselines.

\noindent \textbf{FFHQ Additional Results.} Figures \ref{fig:motion_grid} and \ref{fig:haze_grid} show additional results of our method on the motion deblurring and image dehazing task on the FFHQ dataset. 

\noindent \textbf{ImageNet Additional Results.} Figures \ref{fig:gaussian_imagenet_grid}, \ref{fig:motion_imagenet_grid}, and \ref{fig:high_pass_imagenet_grid} show additional results of our method on the Gaussian deblurring, motion deblurring, and high-pass filter deconvolution tasks on the ImageNet dataset. On the high-pass filter task, we note that there is sometimes a color shift which results from the loss of color information in the measurement. We observe that compared to DPS and a fixed-time frequency curriculum (as in \Cref{fig:ablation_kt}), our method is much more stable and usually gives visually plausible reconstructions instead of strong color artifacts that dominate the reconstruction. 

\begin{figure*}[htb]
    \centering
    \includegraphics[width=\textwidth]{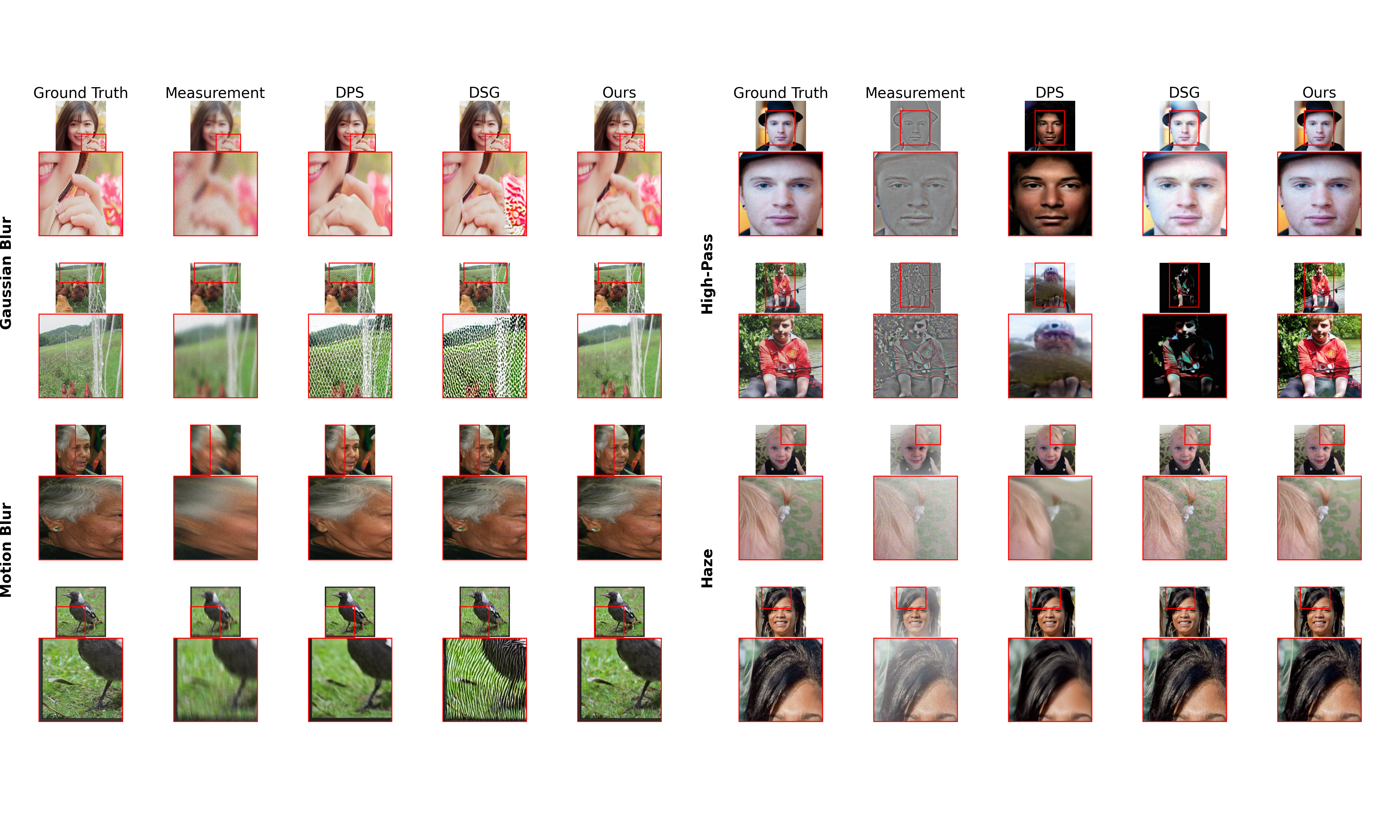}
    \caption{Qualitative results of our method with zoomed in portions of images from \Cref{fig:main_results}. Our method successfully preserves finer details like background pattern.}
    \label{fig:main_results_zoom}
\end{figure*}

\begin{figure*}[htbp]
    \centering
    \includegraphics[width=\textwidth]{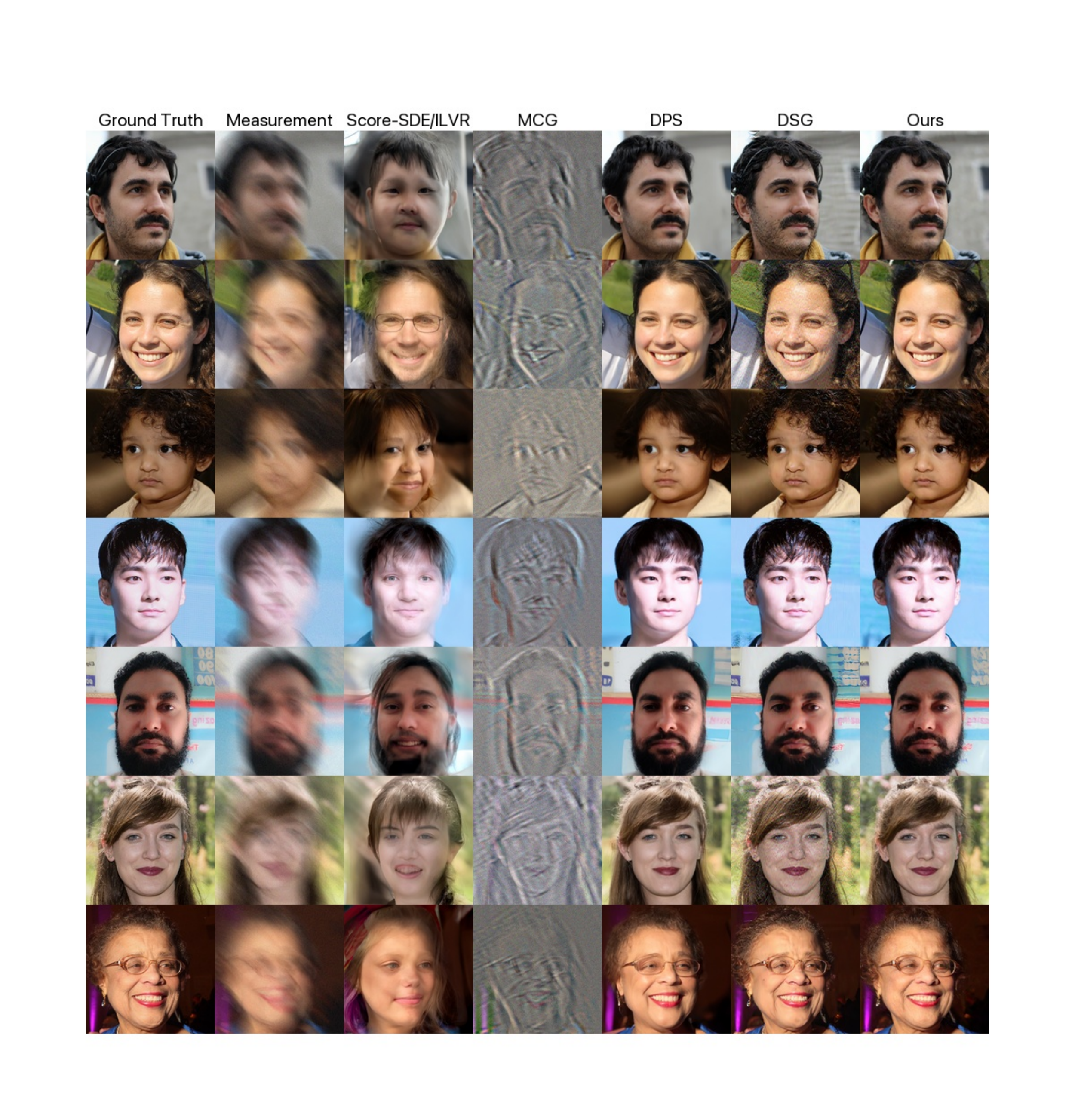}
    \vspace{-10mm}
    \caption{Qualitative motion deblurring results on FFHQ dataset for all baselines we report in \Cref{tab:linear_inv}. The same blur kernel is applied to each image.}
    \label{fig:motion_baseline_grid}
\end{figure*}

\begin{figure*}[htbp]
    \centering
    \includegraphics[width=\textwidth]{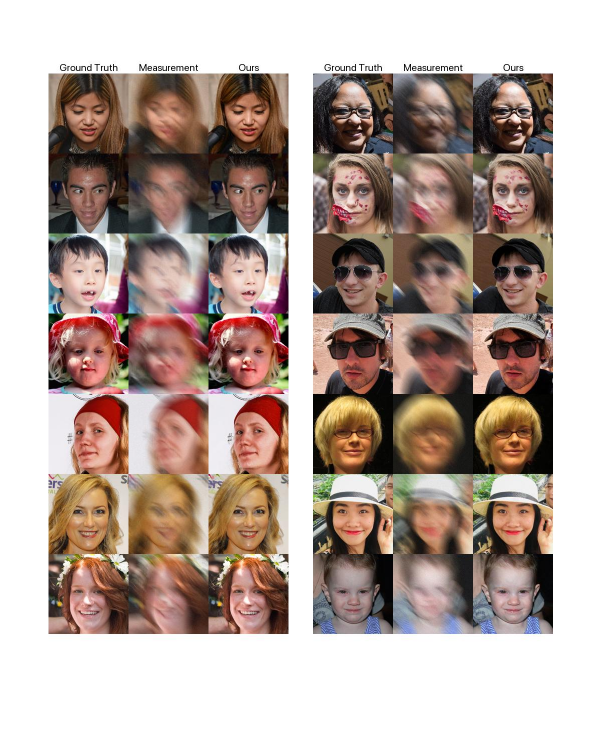}
    \vspace{-25mm}
    \caption{Qualitative motion deblurring results on FFHQ dataset. The same blur kernel is applied to each image.}
    \label{fig:motion_grid}
\end{figure*}

\begin{figure*}[htbp]
    \centering
    \includegraphics[width=\textwidth]{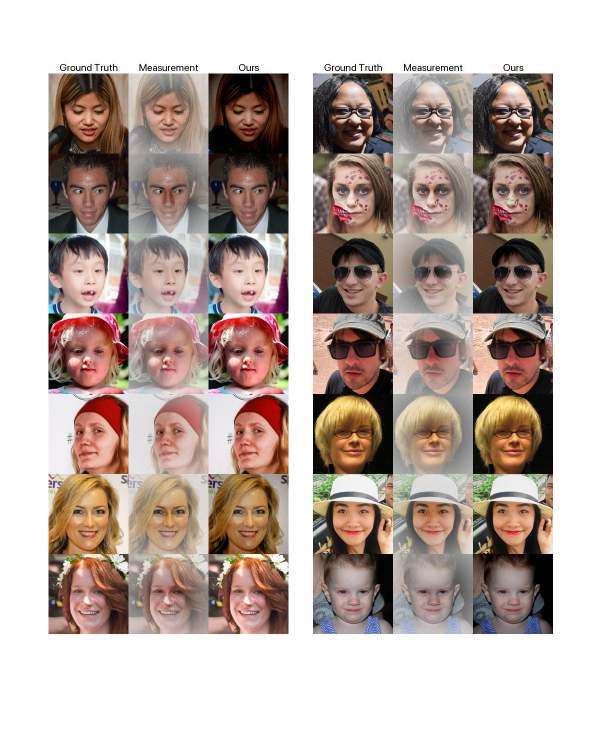}
    \vspace{-25mm}
    \caption{Qualitative image dehazing results on FFHQ dataset.}
    \label{fig:haze_grid}
\end{figure*}

\begin{figure*}[htbp]
    \centering
    \includegraphics[width=\textwidth]{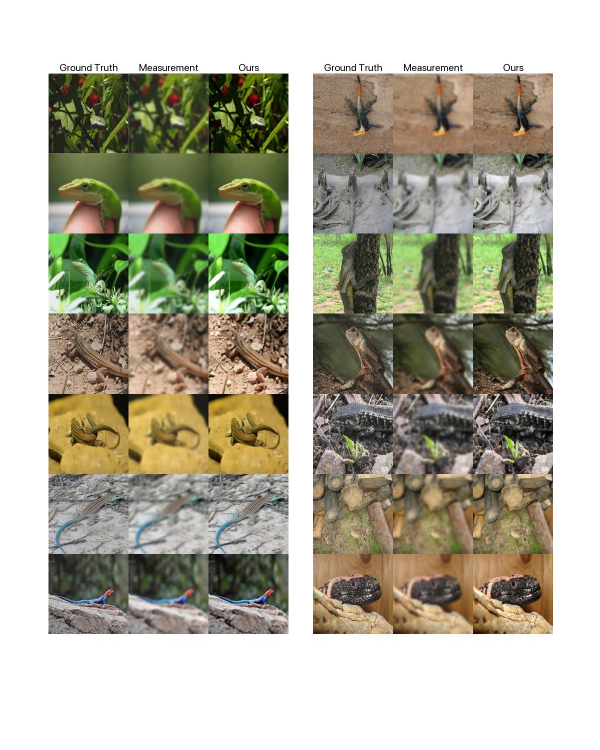}
    \vspace{-25mm}
    \caption{Qualitative Gaussian deblurring results on Imagenet dataset.}
    \label{fig:gaussian_imagenet_grid}
\end{figure*}

\begin{figure*}[htbp]
    \centering
    \includegraphics[width=\textwidth]{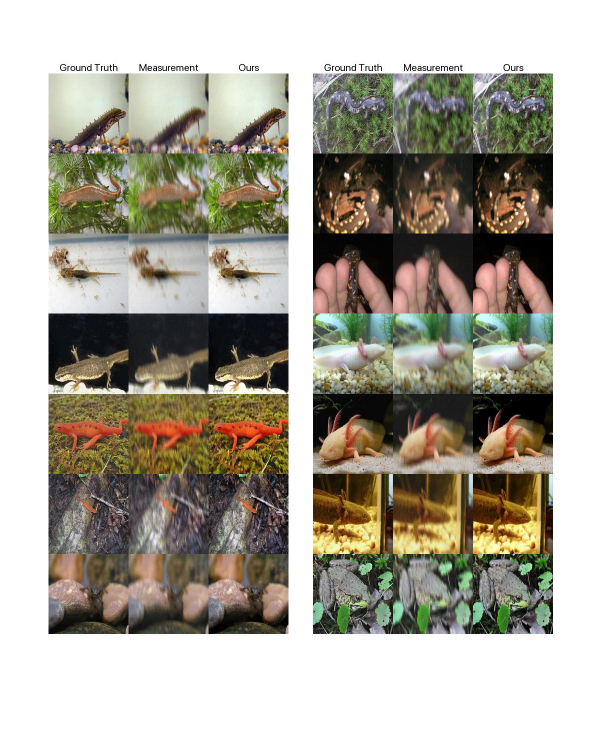}
    \vspace{-25mm}
    \caption{Qualitative motion deblurring results on Imagenet dataset. The same blur kernel is applied to each image.}
    \label{fig:motion_imagenet_grid}
\end{figure*}

\begin{figure*}[htbp]
    \centering
    \includegraphics[width=\textwidth]{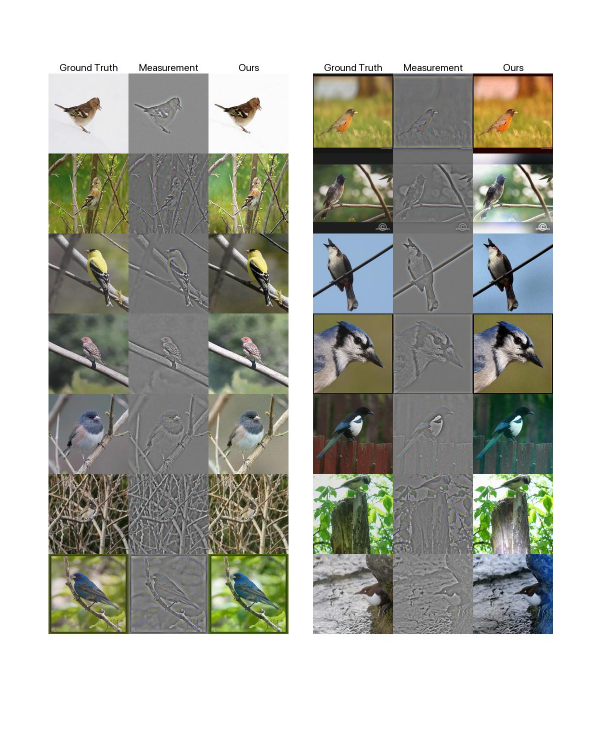}
    \vspace{-25mm}
    \caption{Qualitative results on Imagenet dataset when the measurement is a high-pass filter applied to the ground truth image.}
    \label{fig:high_pass_imagenet_grid}
\end{figure*}

\section{Limitations and Future Work}

While FGPS requires little frequency schedule tuning, the step size still plays a large role in dictating image quality and needs to be carefully tuned. Further, the frequency curriculum is applicable only for image restoration tasks where the measurement is still an image. Lastly, similar to DPS, FGPS requires knowledge of the forward operator during the reverse process, which restricts it to non-blind inverse problems. 

It is important to note that our theoretical results do not paint the full picture of the success of FGPS in practice. Empirically, the Tweedie estimate $\tweedie$ used for the conditional score approximation behaves in complex ways and its frequency structure is not as simple as the form in our theoretical results, $\bm{\Gamma}_t \mbx_0$. We conjecture it is still important to explicitly align the frequency structure of the unconditional score and noisy likelihood score, which is why FGPS outperforms DPS for low-frequency measurements on FFHQ and Imagenet. In addition to the role of the Tweedie estimate, we find an intriguing role of the dataset where FGPS performs even better on harder datasets like ImageNet due to more complex frequency structure. Explaining both these phenomena theoretically is an interesting direction for future work.  That being said, our empirical findings demonstrate that the core idea behind FGPS, aligning the spectral structure of the measurement with the score function, remains effective for complex data. Our curriculum strategy reflects a coarse-to-fine alignment of frequencies, motivated by both the empirical behavior of diffusion models and spectral properties of natural images. We believe this insight opens avenues for more principled guidance mechanisms utilizing the structure of the score function.

Future work would include a rigorous analysis of the step size and how it affects the approximation error. It would also be useful to consider several competing works and their introduced approximation errors using our theoretical analysis as a backbone. Lastly, we hope to extend FGPS to other inverse problems, both blind and non-blind, where the measurement is not an image such as medical imaging tasks.

\end{document}